\documentclass[journal]{IEEEtran}%
\usepackage[center]{caption}
\usepackage{amsmath}
\usepackage{bbm}
\usepackage{graphicx}
\usepackage{epsfig}
\usepackage{array}
\usepackage{psfrag}
\usepackage{amssymb}
\usepackage{mathdots}
\usepackage{color}
\usepackage{pstricks,pst-node,pst-text,pst-3d,pst-plot}
\usepackage{psfrag}
\usepackage{enumerate}
\usepackage{url,cite}
\usepackage{amsfonts}%
\usepackage{amsthm}
\usepackage{dsfont}%
\usepackage{verbatim}%
\usepackage{setspace}
\usepackage{float}
\usepackage{url}
\usepackage{fancyhdr}
\usepackage[hmargin=0.6in , vmargin=1in]{geometry}
\usepackage{algorithm}
\usepackage{algorithmic}
\newtheorem{thm}{Theorem}
\newtheorem{lma}{Lemma}

\newcommand{\gammath}[1]{\gamma_{\rm th}\left( #1 \right)}
\newcommand{\gammathst}[1]{\gamma_{\rm th}^*\left( #1 \right)}
\newcommand{\bfGamma}{{\bf \Gamma}_{\rm th}}
\newcommand{\Gammast}[1]{\bfGamma^{\rm *} \lb #1 \rb}
\newcommand{\bfP}[1]{{\bf P}_{#1}}
\newcommand{\Pst}[1]{{\bf P}_{#1}^*}
\newcommand{\dgamdlamp}[1]{\frac{\partial \gammathst{#1}}{\partial \lambdapst}}
\newcommand{\dUdlamp}[1]{\frac{\partial U_{#1}^*}{\partial \lambdapst}}
\newcommand{\dSdlamp}[1]{\frac{\partial S_{#1}^*}{\partial \lambdapst}}
\newcommand{\dpdlamp}[1]{\frac{\partial p_{#1}^*}{\partial \lambdapst}}

\newcommand{\dUSpdlamp}[1]{\dUdlamp{#1} - \lambdap \dSdlamp{#1} + \lambdad \dpdlamp{#1}}

\newcommand{\Pist}[1]{P_{#1}^* \lb \gammathst{#1} \rb}
\newcommand{\pdf}{f_{\gamma}}
\newcommand{\ccdf}{\bar{F}_{\gamma}}
\newcommand{\lb}{\left (}
\newcommand{\rb}{\right )}
\newcommand{\script}[1]{{\mathcal {#1}}}

\newcommand{\fzb}{f_{z \vert b} \lb z \vert b_i=1 \rb}
\newcommand{\fzf}{f_{z \vert b} \lb z \vert b_i=0 \rb}
\newcommand{\Ibessel}[1]{I_{#1}^{\rm Bes}}

\newcommand{\Pavg}{P_{\rm avg}}
\newcommand{\Iavg}{I_{\rm avg}}
\newcommand{\rmin}{r_{\rm min}}
\newcommand{\pframe}{P_{\rm frame}}
\newcommand{\Dmax}{\bar{D}_{\rm max}}
\newcommand{\invDmaxline}{1/\Dmax}
\newcommand{\invDmax}{\frac{1}{\Dmax}}
\newcommand{\zb}{z_{\rm b}}
\newcommand{\zf}{z_{\rm f}}

\newcommand{\pmd}{P_{\rm MD}}
\newcommand{\pfa}{P_{\rm FA}}
\newcommand{\pskip}[1]{p_{#1}^{\rm skip}}
\newcommand{\Utwomax}{U_2^{\rm max}}
\newcommand{\ptwomax}{p_2^{\rm max}}
\newcommand{\Ts}{T_{\rm s}}

\newcommand{\tinset}{t \in \{1,...t_f\}}

\newcommand{\lambdad}{\lambda_{\rm D}}
\newcommand{\lambdadst}{\lambda_{\rm D}^*}
\newcommand{\lambdap}{\lambda_{\rm P}}
\newcommand{\lambdapst}{\lambda_{\rm P}^*}
\newcommand{\lambdapmin}{\lambdap^{\rm min}}
\newcommand{\lambdapmax}{\lambdap^{\rm max}}
\newcommand{\lambdadmax}{\lambdad^{\rm max}}

\newcommand{\Pmax}{P_{\rm max}}

\newcommand{\lambdai}{\lambda_{\rm I}}

\newcommand{\lambdaist}{\lambda_{\rm I}^*}

\newcommand{\PS}[1]{P_{#1}\lb \gamma \rb}
\newcommand{\PSst}[1]{P^*_{#1}\lb \gamma \rb}
\newcommand{\bfPS}[1]{{\bf P}_{#1}}

\newcommand{\gammaz}[1]{\gamma_{\rm th} \left(#1,z \right)}
\newcommand{\gammazst}[1]{\gamma_{\rm th}^* \left(#1,z \right)}
\newcommand{\GammaS}[1]{\bfGamma \lb #1,z\rb}
\newcommand{\Gammasst}[1]{\bfGamma^{\rm *} \lb #1,z\rb}
\newcommand{\Usoft}{U}
\newcommand{\Ssoft}{S}
\newcommand{\psoft}{p}
\newcommand{\Isoft}{I}
\newcommand{\Usoftst}{U^*}
\newcommand{\Ssoftst}{S^*}
\newcommand{\psoftst}{p^*}
\newcommand{\Isoftst}{I^*}

\begin{document}
\title{Throughput Optimization in Multi-Channel Cognitive Radios with Hard Deadline Constraints}
\author{Ahmed E. Ewaisha, Cihan Tepedelenlio\u{g}lu\\
{School of Electrical, Computer, and Energy Engineering, Arizona State University, Tempe, USA.}\\
\small{Email:\{ewaisha, cihan\}@asu.edu}\\
}
\maketitle
\let\thefootnote\relax\footnote{Copyright (c) 2015 IEEE. Personal use of this material is permitted. However, permission to use this material for any other purposes must be obtained from the IEEE by sending a request to pubs-permissions@ieee.org.

The work in this paper has been supported by NSF Grant CCF-1117041. Parts of this work have appeared in \cite{Ewaisha_Throughput_Maximization}.
}
\begin{abstract}
In a cognitive radio scenario we consider a single secondary user (SU) accessing a multi-channel system. The SU senses the channels sequentially to detect if a primary user (PU) is occupying the channels, and stops its search to access a channel if it offers a significantly high throughput. The optimal stopping rule and power control problem is considered. The problem is formulated as a SU's throughput-maximization problem under a power, interference and packet delay constraints. We first show the effect of the optimal stopping rule on the packet delay, then solve this optimization problem for both the overlay system where the SU transmits only at the spectrum holes as well as the underlay system where tolerable interference (or tolerable collision probability) is allowed. We provide closed-form expressions for the optimal stopping rule, and show that the optimal power control strategy for this multi-channel problem is a modified water-filling approach. We extend the work to multiple SU scenario and show that when the number of SUs is large the complexity of the solution becomes smaller than that of the single SU case. We discuss the application of this problem in typical networks where packets arrive simultaneously and have the same departure deadline. We further propose an online adaptation policy to the optimal stopping rule that meets the packets' hard-deadline constraint and, at the same time, gives higher throughput than the offline policy. 
\end{abstract}

\section{Introduction}
Cognitive Radio (CR) systems are emerging wireless communication systems that allow efficient spectrum utilization \cite{Survey_CR_1st_2006_Akyildiz}. This is because of the use of transceivers that are capable of detecting the presence of licensed (primary) users. The secondary users (SU) use the frequency bands dedicated originally for the primary users (PUs), for their own transmission. Once PU's activity is detected on some frequency channel, the SU refrains from any further transmission on this channel. This may result in service disconnection for the SUs, thus degrading the quality of service (QoS). On the other hand, if the SUs have access to other channels, the QoS can be improved if these channels are efficiently utilized.

The problem of multiple channels in CR systems has gained attention in recent works due to the challenges associated with the sensing and access mechanisms in a multichannel CR system. 
Practical hardware constraints on the SUs' transceivers may prevent them from sensing multiple channels simultaneously to detect the state of these channels (free/busy). This leads the SU to sense the channels sequentially, then decide which channel should be used for transmission \cite{POMDP_Qing_Zhao,Sensing_Order_Poor}. In a time slotted system if sequential channel sensing is employed, the SU senses the channels one at a time and stops sensing when a channel is found free. But due to the independent fading among channels, the SU is allowed to skip a free channel if its quality, measured by its power gain, is low and sense another channel seeking the possibility of a higher future gain. Otherwise, if the gain is high, the SU stops at this free channel to begin transmission. The question of when to stop sensing can be formulated as an optimal stopping rule problem \cite{Sabharwal_NonCR_Multiband, Sensing_Order_Poor, Jia_Multichannel_Tx_Opt_Stop_Rule, Cheng_Simple_Chan_Sensing}. In \cite{Sabharwal_NonCR_Multiband} the authors present the optimal stopping rule for this problem in a non-CR system. The work in \cite{Sensing_Order_Poor} develops an algorithm to find the optimal order by which channels are to be sequentially sensed in a CR scenario, whereas \cite{Jia_Multichannel_Tx_Opt_Stop_Rule} studies the case where the SUs are allowed to transmit on multiple contiguous channels simultaneously. The authors presented the optimal stopping rule for this problem in a non-fading wireless channel. Transmissions on multiple channels simultaneously may be a strong assumption for low-cost transceivers especially when they cannot sense multiple channels simultaneously.

In general, if a perfect sensing mechanism is adopted, the SU will not cause interference to the PU since the former transmits only on spectrum holes (referred to as an overlay system). Nevertheless, if the sensing mechanism is imperfect, or if the SU's system is an underlay one (where the SU uses the channels as long as the interference to the PU is tolerable), the transmitted power needs to be controlled to prevent harmful interference to the PU. References \cite{Pei2013Sensing} and \cite{Adaptive_Rate_Power_CR_Sonia} consider power control and show that the optimal power control strategy is a water-filling approach under some interference constraint imposed on the SU transmitter. Yet, all of the above work studies single channel systems which cannot be extended to multiple channels in a straightforward manner. A multiuser CR system was considered in \cite{Hu_MultiCR_Contention} in a time slotted system. To allocate the frequency channel to one of the SUs, the authors proposed a contention mechanism that does not depend on the SUs' channel gains, thus neglecting the advantage of multiuser diversity. A major challenge in a multichannel system is the sequential nature of the sensing where the SU needs to take a decision to stop and begin transmission, or continue sensing based on the information it has so far. This decision needs to trade-off between waiting for a potentially higher throughput and taking advantage of the current free channel. Moreover, if transmission takes place on a given channel, the SU needs to decide the amount of power transmitted to maximize its throughput given some average interference and average power constraints.


In this work, we model the overlay and underlay scenarios of a multi-channel CR system that are sensed sequentially. The problem is solved for a single SU first then we discuss extensions to a multi-SU scenario. For the single SU case, the problem is formulated as a joint optimal-stopping-rule and power-control problem with the goal of maximizing the SU's throughput subject to average power and average interference constraints. This formulation results in increasing the expected service time of the SU's packets. The expected service time is the average number of time slots that pass while the SU attempts to find a free channel, before successfully transmitting a packet. The increase in the service time is due to skipping free channels, due to their poor gain, hoping to find a future channel of sufficiently high gain. If no channels having a satisfactory gain were found, the SU will not be able to transmit its packet, and will have to wait for longer time to find a satisfactory channel. This increase in service time increases the queuing delay. Thus, we solve the problem subject to a bound on the expected service time which controls the delay (we note that in this paper we use the word delay to refer to the service time). In the multiple SUs case, we show that the solution to the single SU problem can be applied directly to the multi-SU system with a minor modification. We also show that the complexity of the solution decreases when the system has a large number of SUs.

To the best of our knowledge, this is the first work to study the joint power-control and optimal-stopping-rule problem in a multi channel CR system. Our contribution in this work is the formulation of a joint power-control and optimal-stopping-rule problem that also incorporates a delay constraint and present a low complexity solution in the presence of interference/collision constraint from the SU to the PU due to the imperfect sensing mechanism. The preliminary results in \cite{Ewaisha_Throughput_Maximization} consider an overlay framework for single user case while neglecting sensing errors. But in this work, we also study the problem in the underlay scenario where interference is allowed from the secondary transmitter (ST) to the primary receiver (PR) and extend to multiple SU case. We also generalize the solution to the multi-SU case when the number of SUs is large. We discuss the applicability of our formulation in typical delay-constrained scenarios where packets arrive simultaneously and have a same deadline. We show that the proposed algorithm can be used to solve this problem offline, to maximize the throughput and meet the deadline constraint at the same time. Moreover, we propose an online algorithm that gives higher throughput compared to the offline approach while meeting the deadline constraint.

%

The rest of this paper is organized as follows: The overlay system model and the underlying assumptions are presented in Section \ref{Overlay_System_Model}. In Section \ref{Problem_Formulation} the problem is formulated mathematically, the main objective is stated and the solution to the overlay problem is proposed. Then, the underlay system model is discussed in Section \ref{Underlay} and the optimal solution is presented. In Section \ref{Multi_SU} the extension to multiple SUs is discussed.
 In Section \ref{General_Delay}, the delay constraint is generalized to the case where multiple packets arrive at the same time and have the same deadline. An online adaptation solution is also proposed that maximizes the throughput subject to a delay constraint. Finally, numerical results are shown in Section \ref{Results}, while Section \ref{Conclusion} concludes the paper.

Throughout the sequel, we use bold fonts for vectors and asterisk to denote that $x^*$ is the optimal value of $x$; all logarithms are natural, while the expected value operator is denoted $\mathbb{E}[\cdot]$ and is taken with respect to all the random variables in its argument. Finally, we use $(x)^+ \triangleq  {\rm{max}}(x,0)$ and $\mathbb{R}$ to denote the set of the real numbers.


\section{Overlay System Model}
\label{Overlay_System_Model}
Consider a PU network that has a licensed access to $M$ orthogonal frequency channels. Time is slotted with a time slot duration of $\Ts$ seconds.
The SU's network consists of a single ST (SU and ST will be used interchangeably) attempting to send real-time data to its intended secondary receiver (SR) through one of the channels licensed to the PU. Before a time slot begins, the SU is assumed to have ordered the channels according to some sequence (we note that the method of ordering the channels is outside the scope of this work. The reader is referred to \cite{Sensing_Order_Poor} for further details about channel ordering), labeled $1,...,M$. The set of channels is denoted by $\script{M}=\{1,...,M\}$. Before the SU attempts to transmit its packet over channel $i$, it senses this channel to determine its availability ``state'' which is described by a Bernoulli random variable $b_i$ with parameter $\theta_i$ ($\theta_i$ is called the availability probability of channel $i$). If $b_i=0$ (which happens with probability $\theta_i$), then channel $i$ is free and the SU may transmit over it until the on-going time slot ends. If $b_i=1$, channel $i$ is busy, and the SU proceeds to sense channel $i+1$. Channel availabilities are statistically independent across frequency channels and across time slots.

We assume that the SU has limited capabilities in the sense that no two channels can be sensed simultaneously. This may be the case when considering radios having a single sensing module with a fixed bandwidth, so that it can be tuned to only one frequency channel at a time. The reader is referred to \cite{Spectrum_Sensing_via_Kolmogorov_Smirnov_Test}, \cite{Zou2011_CognitiveRelaySelection} and \cite{Robust_Spectrum_Sensing_With_Crowd_Sensors} for detailed information on advanced spectrum sensing techniques. Therefore, at the beginning of a given time slot, the SU selects a channel, say channel $1$, senses it for $\tau$ seconds ($\tau \ll \Ts/M$), and if it is free, it uses it for its own transmission. Otherwise, the SU skips this channel and senses channel $2$, and so on until it finds a free channel. If all channels are busy (i.e. the PU has transmission activities on all $M$ channels) then this time slot will be considered as ``blocked''. In this case, the SU waits for the following time slot and begins sensing following the same channel sensing sequence. As the sensing duration increases, the transmission phase duration decreases which decreases the throughput. But we cannot arbitrarily decrease the value of $\tau$ since this decreases the reliability of the sensing outcome. This trade-off has been studied extensively in the literature, e.g. \cite{Liang2008_SensingThroughputTradeoff}, \cite{Zou2010_AsymptoticOutageProb}. In this work we study the impact of sequential channel sensing on the throughput rather than the sensing duration on the throughput. Hence we assume that $\tau$ is a fixed parameter and is not optimized over. For details on the trade-off between throughput and sensing duration in this sequential sensing problem the reader is referred to \cite{Ewaisha2011Optimization}.

\begin{figure}
	\centering
		\includegraphics[width=0.8\columnwidth]{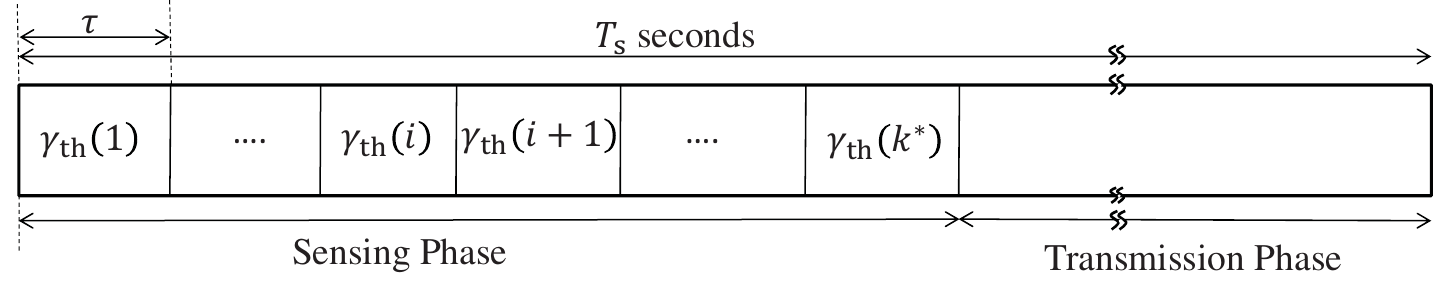}
	\caption{Sensing and transmission phases in one time slot. The SU senses each channel for $\tau$ seconds, determines its state, then probes the gain if the channel is found free. The sensing phase ends if the probed gain $\gamma_i>\gammath{i}$, in which case $k^*=i$. Hence, $k^*$ is a random variable that depends on the channel states and gains.}
	\label{Time_Slot_Fig}
\end{figure}

The fading channel between ST and SR is assumed to be flat fading with independent, identically distributed (i.i.d.) channel gains across the $M$ channels. To achieve higher data rates, the SU adapts its data rate according to the instantaneous power gain of the channel before beginning transmission on this channel. To do this, once the SU finds a free channel, say channel $i$, the gain $\gamma_i$ is probed. The data rate will be proportional to $\log(1+P_i(\gamma_i){\gamma}_i)$, where $P_i(\gamma_i)$ is the power transmitted by the SU at channel $i$ as a function of the instantaneous gain~\cite{Goldsmith_Wireless_Comm}. Fig.~\ref{Time_Slot_Fig} shows a potential scenario where the SU senses $k^*$ channels, skips the first $k^*-1$, and uses the $k^*$th channel for transmission until the end of this on-going time slot. In this scenario the SU ``stops'' at the $k^*$th channel, for some $k^*  \in \script{M}$. Stopping at channel $i$ depends on two factors: 1) the availability of channel $b_i$, and 2) the instantaneous power gain $\gamma_i$. Clearly, $b_i$ and $\gamma_i$  are random variables that change from one time slot to another. Hence, $k^*$, that depends on these two factors, is a random variable. More specifically, it depends on the states $[b_1,...,b_M]$ along with the gains of each channel $[\gamma_1,...,\gamma_M]$. To understand why, consider that the SU senses channel $i$, finds it free and probes its gain $\gamma_i$. If ${\gamma}_i$ is found to be low, then the SU skips channel $i$ (although free) and senses channel $i+1$. This is to take advantage of the possibility that $\gamma_j \gg \gamma_i$ for $j>i$. On the other hand, if $\gamma_i$ is sufficiently large, the SU stops at channel $i$ and begins transmission. In that latter case $k^*=i$. Defining the two random vectors $\underline{b}=[b_1,...,b_M]^T$ and $\underline{\gamma}=[\gamma_1,...,\gamma_M]^T$, $k^*$ is a deterministic function of $\underline{b}$ and $\underline{\gamma}$.

We define the stopping rule by defining a threshold $\gammath{i}$ to which each $\gamma_i$ is compared when the $i$th channel is found free. If $\gamma_i \geq \gammath{i}$, the SU ``stops'' and transmits at channel $i$. Otherwise, channel $i$ is skipped and channel $i+1$ sensed. In the extreme case when $\gammath{i}=0$, the SU will not skip channel $i$ if it is found free. Increasing $\gammath{i}$ allows the SU to skip channel $i$ whenever ${\gamma}_i<\gammath{i}$, to search for a better channel, thus potentially increasing the throughput. Setting $\gammath{i}$ too large allows channel $i$ to be skipped even if $\gamma_i$ is high. This constitutes the trade-off in choosing the thresholds $\gammath{i}$. The optimal values of $\gammath{i}$ $i\in \script{M}$, determine the optimal stopping rule.

Let $P_i(\gamma)$ denote the power transmitted at the $i$th channel when the instantaneous channel gain is $\gamma$, if channel $i$ was chosen for transmission. Since the SU can transmit on one channel at a time, the power transmitted at any time slot at channel $i$ is $P_i(\gamma_i) \mathds{1} \left(i=k^* \right)$, where $\mathds{1} \left(i=k^* \right)=1$ if $i=k^*$ and $0$ otherwise.
Define $c_i \triangleq 1-\frac{i\tau}{\Ts}$ as the fraction of the time slot remaining for the SU's transmission if the SU transmits on the $i$th channel in the sensing sequence. The average power constraint is $\mathbb{E}_{\underline{\gamma},\underline{b}} [c_{k^*} P_{k^*}(\gamma_{k^*})] \leq P_{\rm avg}$, where the expectation is with respect to the random vectors $\underline{\gamma}$ and $\underline{b}$. We will henceforth drop the subscript from the expected value operator ${\mathbb E}$. This expectation can be calculated recursively from
\begin{align}
\nonumber S_i(\bfGamma(i),&\bfP{i})=\theta_i c_i \int_{\gammath{i}}^\infty{P_i(\gamma) f_{\gamma_i}(\gamma) \,d\gamma}+ \\
&\left[1-\theta_i \bar{F}_{\gamma_i}(\gammath{i}) \right]S_{i+1}(\bfGamma(i+1),\bfP{i+1}),
\label{Average_Power}
\end{align}
$i \in \script{M}$, where $\bfP{i} \triangleq [P_i(\gamma),...,P_M(\gamma)]^T$ and $\bfGamma(i) \triangleq [\gammath{i},...,\gammath{M}]^T$ are the vectors of the power functions and thresholds respectively, with $S_{M+1}(\bfGamma(M+1),\bfP{M+1})  \triangleq  0$, $f_{\gamma_i}(\gamma)$ is the Probability Density Function (PDF) of the gain $\gamma_i$ of channel $i$, and $\bar{F}_{\gamma_i}(x) \triangleq \int_x^{\infty}{f_{\gamma_i}(\gamma) \, d\gamma}$ is the complementary cumulative distribution function. The first term in (\ref{Average_Power}) is the average power transmitted at channel $i$ given that channel is chosen for transmission (i.e. given that $k^*=i$). The second term represents the case where channel $i$ is skipped and channel $i+1$ is sensed. It can be shown that $S_1(\bfGamma(1),\bfP{1})=\mathbb{E} \left[c_{k^*} P_{k^*}(\gamma) \right]$. Moreover, we will also drop the index $i$ from the subscript of $f_{\gamma_i}(\gamma)$ and $\bar{F}_{\gamma_i}(\gamma)$ since channels suffer i.i.d. fading. Although we have only included an average power constraint in our problem, we will modify, after solving the problem, the solution to include an instantaneous power constraint as well.

The SU's average throughput is defined as $\mathbb{E} [c_{k^*} \log(1+P_{k^*}(\gamma_{k^*})\gamma_{k^*})]$. Similar to the average power, we denote the expected throughput as $U_1(\bfGamma(1),\bfP{1})$ which can be derived using the following recursive formula
\begin{align}
\nonumber U_i(\bfGamma(i),&\bfP{i})=\theta_i c_i \int_{\gammath{i}}^{\infty}{\log \lb 1+P_i(\gamma)\gamma \rb f_{\gamma}(\gamma)} \, d\gamma +\\
& \left[1-\theta_i \bar{F}_{\gamma}(\gammath{i}) \right] U_{i+1} \lb \bfGamma(i+1),\bfP{i+1} \rb
\label{Reward}
\end{align}
$i\in \script{M}$, with $U_{M+1}(\cdot,\cdot) \triangleq 0$. $U_1(\bfGamma(1),\bfP{1})$ represents the expected data rate of the SU as a function of the threshold vector $\bfGamma(1)$ and the power function vector $\bfP{1}$.

If the SU skips all channels, either due to being busy, due to their low gain or due to a combination of both, then the current time slot is said to be blocked. The SU has to wait for the following time slot to begin searching for a free channel again. This results in a delay in serving (transmitting) the SU's packet. Define the delay $D$ as the number of time slots the SU consumes before successfully transmitting a packet. That is, $D-1$ is a random variable that represents the number of consecutively blocked time slots. In real-time applications, there may exist some average delay requirement $\bar{D}_{\rm{max}}$ on the packets that must not be exceeded. Since the availability of each channel is independent across time slots, $D$ follows a geometric distribution having $\mathbb{E}[D]= \left(\rm{Pr}[\rm{Success}]\right)^{-1}$ where $\rm{Pr}[\rm{Success}]=1-\rm{Pr}[\rm{Blocking}]$. In other words, $\rm{Pr}[\rm{Success}]$ is the probability that the SU finds a free channel with high enough gain so that it does not skip all $M$ channels in a time slot. It is given by $\rm{Pr}[\rm{Success}]  \triangleq  p_1(\bfGamma(1))$ which can be calculated recursively using the following equation
\begin{equation}
\label{Prob_recursive}
p_i(\bfGamma(i))=\theta_i \bar{F}_\gamma(\gammath{i})+ \left[1-\theta_i \bar{F}_\gamma(\gammath{i}) \right]p_{i+1}(\bfGamma(i+1)),
\end{equation}
$i\in \script{M}$, where $p_{M+1} \triangleq 0$. Here, $p_i(\bfGamma(i))$ is the probability of transmission on channel $i$, $i+1$,..., or $M$.

\section{Problem Statement and Proposed Solution}
\label{Problem_Formulation}
From equation (\ref{Reward}) we see that the SU's expected throughput $U_1$ depends on the threshold vector $\bfGamma(1)$ and the power vector $\bfP{1}$. The goal is to find the optimum values of $\bfGamma(1) \in {\mathbb R}^M$ and functions $\bfP{1}$ that maximize $U_1$ subject to an average power constraint and an expected packet delay constraint. The delay constraint can be written as $\mathbb{E}[D] \leq \Dmax$ or, equivalently, $p_1(\bfGamma(1)) \geq \invDmaxline$. Mathematically, the problem becomes

\begin{equation}
\begin{array}{ll}
\rm{maximize}& U_1(\bfGamma(1),\bfP{1})\\
\label{Prob_Opt_Pow_Control}
\rm{subject \; to} &S_1(\bfGamma(1),\bfP{1}) \leq \Pavg\\
& p_1(\bfGamma(1)) \geq \frac{1}{\bar{D}_{\rm{max}}}\\
\rm{variables} & \bfGamma(1),\bfP{1},
\end{array}
\end{equation}
where the first constraint represents the average power constraint, while the second is a bound on the average packet delay. 
We allow the power $P_i$ to be an arbitrary function of $\gamma_i$ and optimize over this function to maximize the throughput subject to average power and delay constraints. Even though (\ref{Prob_Opt_Pow_Control}) is not proven to be convex, we provide closed-form expressions for the optimal thresholds and power-functions vector. To this end, we first calculate the Lagrangian associated with (\ref{Prob_Opt_Pow_Control}). Let $\lambda_{\rm P}$ and $\lambda_{\rm D}$ be the dual variables associated with the constraints in problem (\ref{Prob_Opt_Pow_Control}). The Lagrangian for (\ref{Prob_Opt_Pow_Control}) becomes
\begin{align}
\nonumber L&\left(\bfGamma(1),\bfP{1}, \lambda_{\rm P}, \lambda_{\rm D} \right)=U_1 \left( \bfGamma(1),\bfP{1} \right)-\\
&\lambda_{\rm P} \left( S_1(\bfGamma(1),\bfP{1}) -P_{\rm{avg}}\right)+ \lambda_{\rm D} \left( p_1(\bfGamma(1)) - \frac{1}{\bar{D}_{\rm{max}}} \right).
\label{Lagrange_Optimum_Pow_Control}
\end{align}
Differentiating (\ref{Lagrange_Optimum_Pow_Control}) with respect to each of the primal variables $P_i(\gamma)$ and $\gammath{i}$ and equating the resulting derivatives to zero, we obtain the KKT equations below which are necessary conditions for optimality \cite{Cvx_Boyd}, \cite{Miersemann_Calc_Var}:
\begin{align}
	\label{Water_Filling}
	&P_i^*(\gamma)=\left(\frac{1}{\lambdapst} - \frac{1}{\gamma}\right)^+, \hspace{0.5cm} \gamma> \gammathst{i},\\
	\nonumber &\log \left(1+\left(\frac{1}{\lambdapst}-\frac{1}{\gammathst{i}}\right)^+\gammathst{i} \right) - \lambdapst \left(\frac{1}{\lambdapst} - \frac{1}{\gammathst{i}} \right)^+\\
	&= \frac{U_{i+1}^* - \lambdapst S_{i+1}^*-\lambdadst \cdot \left( 1-p_{i+1}^* \right)}{c_i},
	\label{gamma_i_Equation}\\
\label{Primal_Dual_Feasible}
	&S_1^* \leq P_{\rm{avg}} \hspace{0.2cm}, \hspace{0.2cm} p_1^* \geq \frac{1}{\bar{D}_{\rm{max}}}  \hspace{0.2cm}, \hspace{0.2cm} \lambdapst \geq 0 \hspace{0.2cm}, \hspace{0.2cm} \lambdadst \geq 0,\\
	\label{Comp_Slackness_Power}
	&\lambdapst \cdot \left( S_1^* -P_{\rm{avg}}\right)=0,\\
	\label{Comp_Slackness_Delay}
	&\lambdadst \cdot \left( p_1^* - \frac{1}{\bar{D}_{\rm{max}}} \right)=0,
\end{align}
$i\in \script{M}$. We use $U_{i+1}^* \triangleq U_{i+1}\left(\Gammast{i+1},\Pst{i+1} \right)$ while $S_{i+1}^* \triangleq S_{i+1}\left(\Gammast{i+1},\Pst{i+1} \right)$ and $p_{i+1}^* \triangleq p_{i+1}\left(\Gammast{i+1}\right)$ for brevity in the sequel. We note that $U_{M+1} \lb \cdot , \cdot \rb = S_{M+1} \lb \cdot , \cdot \rb = p_{M+1} \lb \cdot \rb \triangleq 0$ by definition. 
We observe that these KKT equations involve the primal ($\Gammast{1}$ and $\Pst{1}$) and the dual ($\lambdapst$ and $\lambdadst$) variables. Our approach is to find a closed-form expression for the primal variables in terms of the dual variables, then propose a low-complexity algorithm to obtain the solution for the dual variables. The optimality of this approach is discussed at the end of this section (in Section \ref{Optimality_of_Approach}) where we show that, loosely speaking, the KKT equations provide a unique solution to the primal-dual variables. Hence, based on this unique solution, and on the fact that the KKT equations are necessary conditions for the optimal solution, then this solution is not only necessary but sufficient as well, and hence optimal.


\subsection{Solving for Primal Variables}
\label{Primal_Variables}
Equation (\ref{Water_Filling}) is a water-filling strategy with a slight modification due to having the condition $\gamma>\gammath{i}$. This condition comes from the sequential sensing of the channels which is absent in the classic water-filling strategy \cite{Goldsmith_Wireless_Comm}. Equation (\ref{Water_Filling}) gives a closed-form solution for $\bfP{1}$. On the other hand, the entries of the vector $\Gammast{1}$ are found via the set of equations (\ref{gamma_i_Equation}). Note that equation (\ref{gamma_i_Equation}) indeed forms a set of $M$ equations, each solves for one of the $\gammathst{i}$, $i\in \script{M}$. We refer to this set as the ``threshold-finding'' equations. For a given value of $i$, solving for $\gammathst{i}$ requires the knowledge of only $\gammathst{i+1}$ through $\gammathst{M}$, and does not require knowing $\gammathst{1}$ through $\gammathst{i-1}$. Thus, these $M$ equations can be solved using back-substitution starting from $\gammathst{M}$. To solve for $\gammathst{i}$, we use the fact that $\gammathst{i} \geq \lambdapst$ that is proven in the following lemma.
\begin{lma}
\label{Lma_Stochastic_Dominance}
The optimal solution of problem (\ref{Prob_Opt_Pow_Control}) satisfies $\gamma_{\rm th}^*(i) \geq \lambdapst$ $\forall i\in \script{M}$.
\end{lma}
\begin{proof}
See appendix \ref{Apdx_Stochastic_Dominance} for the proof.
\end{proof}
The intuition behind Lemma \ref{Lma_Stochastic_Dominance} is as follows. If, for some channel $i$, $\gammathst{i}<\lambdapst$ was possible, and the instantaneous gain $\gamma_i$ happened to fall in the range $[\gammathst{i},\lambdapst)$ at a given time slot, then the SU will not skip channel $i$ since $\gamma_i>\gammathst{i}$. But the power transmitted on channel $i$ is $P_i(\gamma_i)=\left(1/\lambdapst-1/\gamma_i \right)^+=0$ since $\gamma_i<\lambdapst$. This means that the SU will neither skip nor transmit on channel $i$, which does not make sense from the SU's throughput perspective. To overcome this event, the SU needs to set $\gammathst{i}$ at least as large as $\lambdapst$ so that whenever $\gamma_i<\lambdapst$, the SU skips channel $i$ rather than transmitting with zero power.

Lemma \ref{Lma_Stochastic_Dominance} allows us to remove the $( \cdot )^+$ sign in equation (\ref{gamma_i_Equation}) when solving for $\gammathst{i}$. 
Rewriting (\ref{gamma_i_Equation}) we get
\begin{align}
&\nonumber \frac{-\lambdapst}{\gammathst{i}} \exp \lb \frac{-\lambdapst}{\gammathst{i}} \rb =\\
& -\exp \lb -\frac{U_{i+1}^* - \lambdapst S_{i+1}^*-\lambdadst \cdot \left( 1-p_{i+1}^* \right)}{c_i} - 1 \rb, \hspace{0.1cm} i \in \script{M},
\label{gamma_i_2}
\end{align}
Equation (\ref{gamma_i_2}) is now on the form $W\exp(W)=c$, whose solution is $W=W_0(c)$, where $W_0(x)$ is the principle branch of the Lambert W function \cite{Lambert_W_Function} and is given by $W_0(x)=\sum_{n=1}^{\infty} \frac{\left( -n \right)^{n-1}}{n!}x^n$. The only solution to \eqref{gamma_i_2} which satisfies Lemma \ref{Lma_Stochastic_Dominance} is given for $i\in \script{M}$ by
\begin{equation}
\gammathst{i}=\frac{-\lambdapst}{W_0 \left( -\exp \left(-\frac{\left(U_{i+1}^* - \lambdapst S_{i+1}^*-\lambdadst \left( 1-p_{i+1}^* \right)\right)^+}{c_i}-1\right) \right)}.
\label{Gamma_Solution_Lambert_W}
\end{equation}

Hence, $\Gammast{1}$ and $\Pst{1}$ are found via equations (\ref{Gamma_Solution_Lambert_W}) and (\ref{Water_Filling}) respectively which are one-to-one mappings from the dual variables $(\lambdapst,\lambdadst)$. And if we had an instantaneous power constraint $P_i(\gamma)\leq \Pmax$, we could write down the Lagrangian and solve for $P_i(\gamma)$. The details are similar to the case without an instantaneous power constraint and are, thus, omitted for brevity. The reader is referred to \cite{Adaptive_Rate_Power_CR_Sonia} for a similar proof. The expression for $P_i^*(\gamma)$ is given by
	\begin{equation}
P_i^*(\gamma)=\left \{
\begin{array}{lll}
	\left(\frac{1}{\lambda_{\rm P}^*}-\frac{1}{\gamma} \right)^+ & \mbox{if } \frac{1}{\lambda_{\rm P}^*}-\frac{1}{\gamma} < P_{\rm max} \\
	P_{\rm max} & \mbox{otherwise.}
\end{array}
\right.
\end{equation}

Since the optimal primal variables are explicit functions of the optimal dual variables, once the optimal dual variables are found, the optimal primal variables are found and the optimization problem is solved. We now discuss how to solve for these dual variables.

\subsection{Solving for Dual Variables}
\label{Dual_Variables}
The optimum dual variable $\lambdapst$ must satisfy equation (\ref{Comp_Slackness_Power}). Thus if $\lambdapst>0$, then we need $S_1^* -P_{\rm{avg}}=0$. This equation can be solved using any suitable root-finding algorithm. Hence, we propose Algorithm \ref{Alg_lambda_P_lambda_D} that uses bisection \cite{Numerical_Recipes_Ch9}. In each iteration $n$, the algorithm calculates $S_1^*$ given that $\lambdap=\lambdap^{(n)}$, and given some fixed $\lambdad$, compares it to $\Pavg$ to update $\lambdap^{(n+1)}$ accordingly. The algorithm terminates when $S_1^*=\Pavg$, i.e. $\lambdap^{(n)}=\lambdapst$. The superiority of this algorithm over the exhaustive search is due to the use of the bisection algorithm that does not go over all the search space of $\lambdap$. In order for the bisection to converge, there must exist a single solution for equation $S_1^*=\Pavg$. This is proven in Theorem \ref{Thm_Unique_Solution_S}.
\begin{thm}
\label{Thm_Unique_Solution_S}
$S_1^*$ is decreasing in $\lambdapst \in [0,\infty)$ given some fixed $\lambdadst \geq 0$. Moreover, the optimal value $\lambdapst$ satisfying $S_1^*=\Pavg$ is upper bounded by $\lambdapmax \triangleq \sum_{i=1}^M \theta_i c_i /\Pavg$.
\end{thm}
\begin{proof}
See Appendix \ref{Apdx_Unique_Solution_S} for the proof.
\end{proof}
We note that Algorithm \ref{Alg_lambda_P_lambda_D} can be systematically modified to call any other root-finding algorithm (e.g. the secant algorithm \cite{Numerical_Recipes_Ch9} that converges faster than the bisection algorithm).

\begin{algorithm}
\caption{Finding $\lambdapst$ given some $\lambdad$}
\begin{algorithmic}[1]
\label{Alg_lambda_P_lambda_D}
\STATE Initialize $n \leftarrow 1$, $\lambdapmin \leftarrow 0$, $\lambdapmax \leftarrow \sum_{i=1}^M \theta_i c_i /\Pavg$, $\lambdap^{(1)} \leftarrow \lb \lambdapmin + \lambdapmax \rb /2$
\WHILE{$\vert S_1^*-P_{\rm avg} \vert > \epsilon$}
\STATE Calculate $S_1^*$ given that $\lambdapst=\lambdap^{(n)}$. Call it $S^{(n)}$.
\IF{$S^{(n)}-\Pavg>0$}
\STATE $\lambdapmin=\lambdap^{(n)}$
\ELSE
\STATE $\lambdapmax=\lambdap^{(n)}$
\ENDIF
\STATE $\lambdap^{(n+1)} \leftarrow \lb \lambdapmin + \lambdapmax \rb /2$
\STATE $n \leftarrow n+1$
\ENDWHILE
\STATE $\lambdapst \leftarrow \lambdap^{(n)}$
\end{algorithmic}
\end{algorithm}

Now, to search for $\lambdadst$, we state the following lemma.

\begin{lma}
\label{Lma_lambdadst_Bound}
The optimum value $\lambdadst$ that solves problem (\ref{Prob_Opt_Pow_Control}) satisfies $0 \leq \lambdadst <\lambdadmax$, where
\begin{equation}
\lambdadmax \triangleq \frac{c_1 \left[ \log \lb t\rb - t +1 \right] + \Utwomax}{1-\ptwomax}
\label{lambdadst_Bound}
\end{equation}
with $t \triangleq \lb\min \lb \lambdapmax,\ccdf^{-1} \lb \frac{1}{\theta_1 \Dmax}\rb\rb \rb /\lb\ccdf^{-1} \lb \frac{1}{\theta_1 \Dmax}\rb\rb$ and $\Utwomax$ is an upper bound on $U_2^*$ and is given by $\lb\int_{\lambdapmax}^\infty \log \lb \gamma/\lambdapmax\rb \pdf(\gamma) \, d\gamma \rb \lb \sum_{i=2}^M \theta_i c_i\rb$, while $\ptwomax$ is an upper bound on $p_2^*$ and is given by $\sum_{i=2}^M \prod_{j=2}^{i-1} \lb 1-\theta_j \rb \theta_i$.
\end{lma}
\begin{proof}
See Appendix \ref{Apdx_Lma_lambdadst_Bound}.
\end{proof}
Lemma \ref{Lma_lambdadst_Bound} gives an upper bound on $\lambdadst$. This bound decreases the search space of $\lambdadst$ drastically instead of searching over $\mathbb{R}$. Thus the solution of problem (\ref{Prob_Opt_Pow_Control}) can be summarized on 3 steps: 1) Fix $\lambdadst \in [0,\lambdadmax)$ and find the corresponding optimum $\lambdapst$ using Algorithm \ref{Alg_lambda_P_lambda_D}. 2) Substitute the pair $(\lambdapst,\lambdadst)$ in equations (\ref{Water_Filling}) and (\ref{Gamma_Solution_Lambert_W}) to get the power and threshold functions, then evaluate $U_1^*$ from (\ref{Reward}). 3) Repeat steps 1 and 2 for other values of $\lambdadst$ until reaching the optimum $\lambdadst$ that satisfies $p_1^*=\invDmaxline$. If there are multiple $\lambdadst$'s satisfying $p_1^*=\invDmaxline$, then the optimum one is the one that gives the highest $U_1^*$.

Although the order by which the channels are sensed is assumed fixed, the proposed algorithm can be modified to optimize over the sensing order by a relatively low complexity sorting algorithm. Particularly, the dynamic programming proposed in \cite{Sensing_Order_Poor} can be called by Algorithm \ref{Alg_lambda_P_lambda_D} to order the channels. The complexity of the sorting algorithm alone is $O(2^M)$ compared to the $O(M!)$ of the exhaustive search to sort the $M$ channels. The modification to our proposed algorithm would be in step 3 of Algorithm \ref{Alg_lambda_P_lambda_D}, where $S_1^*$ would be optimized over the number of channels (as well as $\Gammast{1}$).

\subsection{Optimality of the Proposed Solution}
\label{Optimality_of_Approach}
Since the problem in (\ref{Prob_Opt_Pow_Control}) is not proven to be convex, the KKT conditions provide only necessary conditions for optimality and need not be sufficient \cite{LinNonlinProg_Luenberger}. This means that there might exist multiple solutions (i.e. multiple solutions for the primal and/or dual variables) satisfying the KKT conditions, at least one of which is optimal. But since Theorem \ref{Thm_Unique_Solution_S} proves that there exists one unique solution to $\lambdapst$ given $\lambdadst$, then $\Gammast{1}$ and $\Pst{1}$ are unique as well (from equations (\ref{Water_Filling}) and (\ref{Gamma_Solution_Lambert_W})) given some $\lambdadst$. Hence, by sweeping $\lambdadst$ over $[0,\lambdadmax)$, we have a unique solution satisfying the KKT conditions, which means that the KKT conditions are sufficient as well and our approach is optimal for problem (\ref{Prob_Opt_Pow_Control}).

\section{Underlay System}
\label{Underlay}
In the overlay system, the SU tries to locate the free channels at each time slot to access these spectrum holes without interfering with the PUs. Recently, the FCC has allowed the SUs to interfere with the PU's network as long as this interference does not harm the PUs \cite{FCC_Interference_Threshold}. If the interference from the SU measured at the PU's receiver is below the tolerable level, then the interference is deemed acceptable.

In order to model the interference at the PR, we assume that the SU uses a channel sensing technique that produces the sufficient statistic $z_i$ at channel $i$ \cite{Spectrum_Sensing_Survey_Poor, Spectrum_Sensing_Survey}. The SU is assumed to know the distribution of $z_i$ given channel $i$ is free and busy, namely $f_{z \vert b} \lb z_i \vert b_i=0 \rb$ and $f_{z \vert b} \lb z_i \vert b_i=1 \rb$ respectively. For brevity, we omit the subscript $i$ from $b_i$ whenever it is clear from the context.
The value of $z_i$ indicates how confident the SU is in the presence of the PU at channel $i$. Thus the SU stops at channel $i$ according to how likely busy it is and how much data rate it will gain from this channel (i.e. according to $z_i$ and $\gamma_i$ respectively). Hence, when the SU senses channel $i$ to acquire $z_i$, the channel gain $\gamma_i$ is probed and compared to some function $\gamma_{\rm th}(i,z_i)$; if $\gamma_i\geq\gamma_{\rm th}(i,z_i)$ transmission occurs on channel $i$, otherwise, channel $i$ is skipped and $i+1$ is sensed. Potentially, $\gamma_{\rm th}(i,z_i)$ is a function in the statistic $z_i$. This means that, at channel $i$, for each possible value that $z_i$ might take, there is a corresponding threshold $\gammaz{i}$. Before formulating the problem we note that this model can capture the overlay with sensing errors model as a special case where $\fzb=(1-\pmd)\delta(z-\zb)+\pmd\delta(z-\zf)$ while $\fzf=\pfa\delta(z-\zb)+(1-\pfa)\delta(z-\zf)$, where $\pmd$ and $\pfa$ are the probabilities of missed-detection and false-alarm respectively, while $\delta(z)$ is the Dirac delta function, and $\zb$ and $\zf$ that represent the values that $z$ takes when the channel is busy and free, respectively. Hence, the interference constraint, which will soon be described, can be modified to a detection probability constraint and/or a false alarm probability constraint.

The SU's expected throughput is given by $\Usoft_1(\GammaS{1},\bfPS{1})$ which can be calculated recursively from
\begin{equation}
\begin{array}{ll}
\Usoft_i&(\GammaS{i},\bfPS{i})=\\
&c_i \int_{-\infty}^{\infty}{\int_{\gammaz{i}}^{\infty}{\log(1+\PS{i}\gamma) f_{\gamma}(\gamma)} \, d\gamma f_z(z)} \, dz +
\\ &\pskip{i} \Usoft_{i+1}(\GammaS{i+1},\bfPS{i+1}), \hspace{0.2in} i\in \script{M},
\label{Reward_Soft}
\end{array}
\end{equation}
where $\Usoft_{M+1}(\GammaS{M+1},\bfPS{M+1}) \triangleq 0$, $\GammaS{i}\triangleq [\gammaz{i},...,\gammaz{M}]^T$, $f_z(z) \triangleq \theta_i \fzf + (1-\theta_i) \fzb$ is the PDF of the random variable $z_i$ and $\pskip{i} \triangleq \int_{-\infty}^{\infty}{\int_0^{\gammaz{i}}{f_{\gamma}(\gamma)} \, d\gamma f_z(z)} \, dz$. The first term in (\ref{Reward_Soft}) is the SU's throughput at channel $i$ averaged over all realizations of $z_i$ and that of $\gamma_i \geq \gammaz{i}$. The second term is the average throughput when the SU skips channel $i$ due to its low gain. Also, let the average interference from the SU's transmitter to the PU's receiver, aggregated over all $M$ channels, be $\Isoft_1(\GammaS{1},\bfPS{1})$. This represents the total interference affecting the PU's network due to the existence of the SU. The SU is responsible for guaranteeing that this interference does not exceed a threshold $\Iavg$ dictated by the PU's network. $\Isoft_1(\GammaS{1},\bfPS{1})$ can be derived using the following recursive formula
\begin{equation}
\begin{array}{ll}
&\Isoft_i(\GammaS{i},\bfPS{i})=\\
&\left( 1- \theta_i \right) c_i \int_{-\infty}^{\infty}{\int_{\gammaz{i}}^{\infty}{\PS{i} f_{\gamma}(\gamma)} \, d\gamma f_{z \vert b} \lb z \vert b_i=1 \rb} \, dz \\
&+ \pskip{i}\Isoft_{i+1}(\GammaS{i+1},\bfPS{i+1}), \hspace{0.3in} i\in \script{M},
\label{Interference_Soft}
\end{array}
\end{equation}
where $\Isoft_{M+1}(\GammaS{M+1},\bfPS{M+1}) \triangleq 0$. This interference model is based on the assumption that the channel gain from the SU's transmitter to the PU's receiver is known at the SU's transmitter. This is the case for reciprocal channels when the PR acts as a transmitter and transmits training data to its intended primary transmitter (when it is acting as a receiver)~\cite{Wireless_Comm_Tse}. The ST overhears this training data and estimates the channel from itself to the PR. Moreover, the gain at each channel from the ST to the PR is assumed unity for presentation simplicity. This could be extended easily to the case of non-unity-gain by multiplying the first term in (\ref{Interference_Soft}) by the gain from the ST to the PR at channel $i$. Finally, $\psoft_1(\GammaS{1})$ is the probability of a successful transmission in the current time slot and can be calculated using
\begin{equation}
\begin{array}{ll}
	\psoft_i(\GammaS{i})=&\int_{-\infty}^{\infty}{\int_{\gammaz{i}}^{\infty}{f_{\gamma}(\gamma)} \, d\gamma f_z(z)} \, dz +\\
	&\pskip{i}\psoft_{i+1}(\GammaS{i+1}),
\label{Prob_recursive_Soft}
\end{array}
\end{equation}
$i\in \script{M}$, $\psoft_{M+1}(\GammaS{M+1}) \triangleq 0$. Given this background, the problem is
\begin{equation}
\begin{array}{ll}
\rm{maximize}& \Usoft_1(\GammaS{1},\bfPS{1})\\
\label{Prob_Opt_Pow_Control_Soft}
\rm{subject \; to} &\Isoft_1(\GammaS{1},\bfPS{1}) \leq \Iavg\\
&\psoft_1(\GammaS{1}) \geq \invDmax\\
\rm{variables} & \GammaS{1},\bfPS{1},
\end{array}
\end{equation}
Let $\lambdai$ and $\lambdad$ be the Lagrange multipliers associated with the interference and delay constraints of problem (\ref{Prob_Opt_Pow_Control_Soft}), respectively. 
Problem (\ref{Prob_Opt_Pow_Control_Soft}) is more challenging compared to the overlay case. This is because, unlike in (\ref{Prob_Opt_Pow_Control}), the thresholds in (\ref{Prob_Opt_Pow_Control_Soft}) are functions rather than constants. The KKT conditions for \eqref{Prob_Opt_Pow_Control_Soft} are given by
\begin{align}
\label{Water_Filling_Soft}
	&\PSst{i}=\left( \frac{1}{\lambdaist {\rm Pr}\left[b_i=1 \vert z \right]} - \frac{1}{\gamma}\right)^+ \hspace{0.5cm},\hspace{0.5cm} i\in \script{M}.\\
	\label{gammaS_i}
\nonumber &\gammazst{i}=\\
&\frac{-\lambdaist {\rm Pr}\left[b_i=1 \vert z \right]}{W_0 \left( -\exp \left(-\frac{\left(\Usoftst_{i+1} - \lambdaist \Isoftst_{i+1} - \lambdadst \left(1-\psoftst_{i+1} \right) \right)^+}{c_i}-1\right) \right)}, \hspace{0.25cm} i\in \script{M},
\end{align}
in addition to the primal feasibility, dual feasibility and the complementary slackness equations given in \eqref{Primal_Dual_Feasible}, \eqref{Comp_Slackness_Power} and \eqref{Comp_Slackness_Delay}, where $\Usoftst_{i+1} \triangleq \Usoft_1  \left( \Gammasst{1},\PSst{1} \right)$, $\Isoftst_{i+1} \triangleq \Isoft_1  \left( \Gammasst{1},\PSst{1} \right)$ and $\psoftst_{i+1} \triangleq \psoft_1  \lb \Gammasst{1}\rb$ while ${\rm Pr}\left[b_i=1 \vert z \right]$ is the conditional probability that channel $i$ is busy given $z_i$ and is given by
\begin{equation}
\label{Cond_Prob_i_given_z_i}
{\rm Pr}\left[b_i=1 \vert z \right]=\frac{\left( 1-\theta_i \right) f_{z \vert b} \lb z \vert b_i=1 \rb}{f_z \left( z \right)}.
\end{equation}
Note that $\PSst{i}$ is increasing in $\gamma$ and is upper bounded by the term $1 / \lb \lambdaist {\rm Pr}\left[b_i=1 \vert z \right] \rb$. Hence, as ${\rm Pr}\left[b_i=1 \vert z \right]$ increases, the SU's maximum power becomes more limited, i.e. the maximum power decreases. This is because the PU is more likely to be occupying channel $i$. Thus the power transmitted from the SU should decrease in order to protect the PU.


Algorithm \ref{Alg_lambda_P_lambda_D} can also be used to find $\lambdaist$. Only a single modification is required in the algorithm which is that $S_1^*$ would be replaced by $\Isoftst_1$. Thus the solution of problem (\ref{Prob_Opt_Pow_Control_Soft}) can be summarized on 3 steps: 1) Fix $\lambdadst \in \mathbb{R}^+$ and find the corresponding optimum $\lambdaist$ using the modified version of Algorithm \ref{Alg_lambda_P_lambda_D}. 2) Substitute the pair $(\lambdaist,\lambdadst)$ in equations (\ref{Water_Filling_Soft}) and (\ref{gammaS_i}) to get the power and threshold functions, then evaluate $\Usoftst_1$ from (\ref{Reward_Soft}). 3) Repeat steps 1 and 2 for other values of $\lambdadst$ until reaching the optimum $\lambdadst$ that satisfies $\psoftst_1=\invDmaxline$ and if there are multiple $\lambdadst$'s satisfying $p_1^*=\invDmaxline$, then the optimum one is the one that gives the highest $\Usoftst_1$. This approach yields the optimal solution. Next, Theorem \ref{Thm_Unique_Solution_I} asserts the monotonicity of $\Isoftst_1$ in $\lambdaist$ which allows using the bisection to find $\lambdaist$ given some fixed $\lambdadst$.
\begin{thm}
\label{Thm_Unique_Solution_I}
$\Isoftst_1$ is decreasing in $\lambdaist \in [0,\infty)$ given some fixed $\lambdadst \geq 0$.
\end{thm}
\begin{proof}
We differentiate $\Isoftst_1$ with respect to $\lambdaist$ given that $\PSst{i}$ and $\gammazst{i}$ are given by equations \eqref{Water_Filling_Soft} and \eqref{gammaS_i} respectively, then show that this derivative is negative. The proof is omitted since it follows the same lines of Theorem \ref{Thm_Unique_Solution_S}. 
\end{proof}

Although the interference power constraint is sufficient for the problem to prevent the power functions from going to infinity, in some applications one may have an additional power constraint on the SUs. Hence, problem (\ref{Prob_Opt_Pow_Control_Soft}) can be modified to introduce an average power constraint that is given by $\Ssoft_1(\GammaS{1},\bfPS{1}) \leq \Pavg$ where
\begin{equation}
\begin{array}{ll}
\Ssoft_i(\GammaS{i},\bfPS{i})&=c_i \int_{-\infty}^{\infty}{\int_{\gammaz{i}}^{\infty}{\PS{i} f_{\gamma}(\gamma)} \, d\gamma f_z(z)} \, dz \\
&+ \pskip{i}\Ssoft_{i+1}(\GammaS{i+1},\bfPS{i+1}).
\label{Avg_Pow_Soft}
\end{array}
\end{equation}
It can be easily shown that the solution to the modified problem is similar to that presented in equations (\ref{Water_Filling_Soft}) and (\ref{gammaS_i}) which is
\begin{align}
\label{Water_Filling_Soft_Avg_Pow}
	&\PSst{i}=\left( \frac{1}{\lambdapst +\lambdaist {\rm Pr}\left[b_i=1 \vert z \right]} - \frac{1}{\gamma}\right)^+ \hspace{0.5cm},\\
	\label{gammaS_i_Avg_Pow}
\nonumber &\gammazst{i}=\\
&\frac{-\lb \lambdapst + \lambdaist {\rm Pr}\left[b_i=1 \vert z \right] \rb}{W_0 \left( -\exp \left(-\frac{\left(\Usoftst_{i+1} - \lambdaist \Isoftst_{i+1} -\lambdapst \Ssoftst_{i+1} - \lambdadst \left(1-\psoftst_{i+1} \right) \right)^+}{c_i}-1\right) \right)},
\end{align}
$\forall i\in \script{M}$ where $\Ssoft_i^* \triangleq \Ssoft_i(\Gammasst{i},\PSst{i})$. This solution is more general since it takes into account both the average interference and the average power constraint besides the delay constraint. Moreover, it allows for the case where the power constraint is inactive which happens if the PU has a strict average interference constraint. In this case the optimum solution would result in $\lambdapst=0$ making equations (\ref{Water_Filling_Soft_Avg_Pow}) and (\ref{gammaS_i_Avg_Pow}) identical to equations (\ref{Water_Filling_Soft}) and (\ref{gammaS_i}) respectively.

\section{Multiple Secondary Users}
\label{Multi_SU}
In this section, we show how our single SU framework can be extended to multiple SUs in a multiuser diversity framework without increase in the complexity of the algorithm. We will show that when the number of SUs is high, with slight modifications to the definitions of the throughput, power and probability of success, the single SU optimization problem in \eqref{Prob_Opt_Pow_Control} (or \eqref{Prob_Opt_Pow_Control_Soft}) can capture the multi-SU scenario. Moreover, the proposed solution for the overlay model still works for the multi-SU scenario. Finally, at the end of this section, we show that the proposed algorithm provides a throughput-optimal and delay-optimal solution with even a lower complexity for finding the thresholds compared to the single SU case, if the number of SUs is large.

Consider a CR network with $L$ SUs associated with a centralized secondary base station (BS) in a downlink overlay scenario. Before describing the system model, we would like to note that when we say that channel $i$ will be sensed, this means that each user will independently sense channel $i$ and feedback the sensing outcome to the BS to make a global decision. Although we neglect sensing errors in this section, the analysis will work similarly in the presence of sensing errors by using the underlay model. At the beginning of each time slot the $L$ SUs sense channel 1. If it is free, each SU observes it free with no errors and probes the instantaneous channel gain and feeds it back to the BS. The BS compares the maximum received channel gain among the $L$ received channel gains to $\gammath{1}$. Channel 1 is assigned to the user having the maximum channel gain if this maximum gain is higher than $\gammath{1}$, while the remaining $L-1$ users continue to sense channel 2. On the other hand if the maximum channel gain is less than $\gammath{1}$, channel 1 is skipped and channel 2 is sensed by all $L$ users. 
Unlike the case in the single SU scenario where only a single channel is claimed per time slot, in this multi-SU system, the BS can allocate more than one channel in one time slot such that each SU is not allocated more than one channel and each channel is not allocated to more than one SU. 
Based on this scheme, the expected per-SU throughput $U_1^L$ is calculated from
\begin{align}
\nonumber U_i^l=&\frac{\theta_i c_i}{l} \int_{\gammath{i}}^{\infty}{\log \lb 1+P_i(\gamma)\gamma \rb f_l(\gamma)} \, d\gamma+\\
 &\theta_i \bar{F}_l(\gammath{i})\left(1-\frac{1}{l} \right) U_{i+1}^{l-1} + \left(1-\theta_i \bar{F}_l(\gammath{i}) \right) U_{i+1}^l
\label{Reward_l_SUs}
\end{align}
$i\in \script{M}$ and $l \in\{L-i+1,...,L\}$ with initialization $U_{M+1}^l=0$. Here $f_l(\gamma)$ represents the density of the maximum gain among $l$ i.i.d. users' gains, while $\bar{F}_l(\gamma)$ is its complementary cumulative distribution function. We study the case where $L \gg M$, thus when a channel is allocated to a user we can assume that the remaining number of users is still $L$. Thus we approximate $l$ with $L$ $\forall l \in\{L-i,...,L\}$ and $\forall i \in \script{M}$. 
Similar to the the throughput derived in \eqref{Reward_l_SUs}, we could write the exact expressions for the per-SU average power and per-SU probability of transmission. And since $L \gg M$, we can approximate $S_i^l$ with $S_i^L$ and $p_i^l$ with $p_i^L$, $\forall l \in\{L-i+1,...,L\}$ and $\forall i \in \script{M}$. The per-SU expected throughput $U_1^L$, the average power $S_1^L$ and the probability of transmission $p_1^L$ can be derived from
\begin{align}
\nonumber U_i^L(\bfGamma(i),\bfP{i})=&\frac{\theta_i c_i}{L} \int_{\gammath{i}}^{\infty}{\log \lb 1+P_i(\gamma)\gamma \rb f_{L}(\gamma)} \, d\gamma +\\
&\left[1-\frac{\theta_i \bar{F}_{L}(\gammath{i})}{L} \right] U_{i+1}^L \lb \bfGamma(i+1),\bfP{i+1} \rb
\label{Reward_L_SU}\\
\nonumber S_i^L(\bfGamma(i),\bfP{i})=&\frac{\theta_i c_i}{L} \int_{\gammath{i}}^\infty{P_i(\gamma) f_{L}(\gamma) \,d\gamma}+ \\
&\left[1-\frac{\theta_i \bar{F}_{L}(\gammath{i})}{L} \right]S_{i+1}^L(\bfGamma(i+1),\bfP{i+1}),
\label{Average_Power_L_SU}\\
\nonumber p_i^L(\bfGamma(i))=&\frac{\theta_i}{L} \bar{F}_{L}(\gammath{i})+ \\
&\left[1-\frac{\theta_i \bar{F}_{L}(\gammath{i})}{L} \right]p_{i+1}^L(\bfGamma(i+1)),
\label{Prob_Trans_L_SU}
\end{align}
$i\in \script{M}$, respectively, with  $U_{M+1}^L=S_{M+1}^L=p_{M+1}^L=0$. To formulate the multi-SU optimization problem, we replace $U_1$, $S_1$ and $p_1$  in (\ref{Prob_Opt_Pow_Control}) with $U_1^L$, $S_1^L$ and $p_1^L$ derived in equations (\ref{Reward_L_SU}), (\ref{Average_Power_L_SU}) and (\ref{Prob_Trans_L_SU}), respectively. Taking the Lagrangian and following the same procedure as in Section \ref{Problem_Formulation}, we reach at the solution for $P_i^*$ and $\gammathst{i}$ as given by equations (\ref{Water_Filling}) and (\ref{Gamma_Solution_Lambert_W}) respectively. Hence, equations (\ref{Water_Filling}) and (\ref{Gamma_Solution_Lambert_W}) represent the optimal solution for the multi-SU scenario. The details are omitted since they follow those of the single SU case discussed in Section \ref{Problem_Formulation}.

Next we show that this solution is optimal with respect to the delay as well as the throughput when $L$ is large. We show this by studying the system after ignoring the delay constraint and show that the resulting solution of this system (which is what we refer to as the unconstrained problem) is a delay optimal one as well. The solution of the unconstrained problem is given by setting $\lambdadst=0$ in \eqref{Gamma_Solution_Lambert_W} arriving at
\begin{equation}
\gammathst{i}\vert_{\lambdadst=0}=\frac{-\lambdapst}{W_0 \left( -\exp \left(-\frac{\left(U_{i+1}^{L*} - \lambdapst S_{i+1}^{L*}\right)^+}{c_i}-1\right) \right)}.
\label{Gamma_Solution_Lambert_W_Thr_Opt}
\end{equation}
$\forall i \in \script{M}$. As the number of SUs increases, the per-user expected throughput $U_1^L$ decreases since these users share the total throughput. Moreover, $U_i^L$ decreases as well $\forall i \in \script{M}$ decreasing the value of $\gammathst{i}$ (from equation \eqref{Gamma_Solution_Lambert_W_Thr_Opt} until reaching its minimum (i.e. $\gammathst{i}=\lambdapst$) (the right-hand-side of \eqref{Gamma_Solution_Lambert_W_Thr_Opt} is minimum when its denominator is as much negative as possible. That is, when $W_0(x)=-1$ since $W_0(x)\geq -1$, $\forall x\in \mathbb{R}$) as $L \rightarrow \infty$. From \eqref{Prob_Trans_L_SU}, it can be easily shown that $p_1^L(\bfGamma(1))$ is monotonically decreasing in $\gammath{i}$ $\forall i\in \script{M}$. Thus the minimum possible average delay (corresponding to the maximum $p_1^L(\bfGamma(1))$) occurs when $\gammath{i}$ is at its minimum possible value for all $i \in \script{M}$. Consequently, having $\gammathst{i}=\lambdapst$ means that the system is at the optimum delay point. That is, the unconstrained problem cannot achieve any smaller delay with an additional delay constraint. Hence, the multi-SU problem, that is formulated by adding a delay constraint to the unconstrained problem, achieves the optimum delay performance when $L$ is asymptotically large.

Recall that the overall complexity of solution for the single SU case is due to three factors: 1) evaluating the Lambert W function in Algorithm \ref{Alg_lambda_P_lambda_D}, 2) the bisection algorithm in Algorithm \ref{Alg_lambda_P_lambda_D} and 3) the search over $\lambdad$. On the other hand, the complexity of solution for the multi-SU case decreases asymptotically (as $L\rightarrow \infty$). This is because of two reasons: 1) When $L \gg M$, $\gammathst{i} \rightarrow \lambdapst \forall i \in \script{M}$. Which means that we will not have to evaluate the Lambert W function in \eqref{Gamma_Solution_Lambert_W} but instead we set $\gammathst{i}=\lambdapst$, since $L\gg M$. 2) When $\gammathst{i} = \lambdapst$ there will be no need to find $\lambdadst$ since the delay is minimum (we recall that in the single SU case, we need to calculate $\lambdadst$ to substitute it in \eqref{Gamma_Solution_Lambert_W} to evaluate $\gammathst{i}$, but in the multi-SU case $\gammathst{i} = \lambdapst$).

\section{Generalization of Deadline Constraints}
\label{General_Delay}
In the overlay and underlay schemes discussed thus far, we were assuming that each packet has a hard deadline of one time slot. If a packet is not delivered as soon as it arrives at the ST, then it is dropped from the system. But in real-time applications, data arrives at the ST's buffer on a frame-by-frame structure. Meaning multiple packets (constituting the same frame) arrive simultaneously rather than one at a time. A frame consists of a fixed number of packets, and each packet fits into exactly one time slot of duration $\Ts$. Each frame has its own deadline and thus, packets belonging to the same frame have the same deadline \cite{Adaptive_NC_Deadline}. This deadline represents the maximum number of time slots that the packets, belonging to the same frame, need to be transmitted by, on average.

In this section we solve this problem for the overlay scenario. The solution presented in Section \ref{Problem_Formulation} can be thought of as a special case of the problem presented in this section where the deadline was equal to $1$ time slot and each frame consists of one packet. We show that the solution presented in Section \ref{Problem_Formulation} can be used to solve this generalized problem in an offline fashion (i.e. before attempting to transmit any packet of the frame). Moreover, we propose an online update algorithm that updates the thresholds and power functions each time slot and show that this outperforms the offline solution.


\subsection{Offline Solution}
Assume that each frame consists of $K$ packets and that the entire frame has a deadline of $t_f$ time slots ($t_f>K$). If the SU does not succeed in transmitting the $K$ packets before the $t_f$ time slots, then the whole frame is considered wasted. Since instantaneous channel gains and PU's activities are independent across time slots, the probability that the SU succeeds in transmitting the frame in $t_f$ time slots or less is given by
\begin{equation}
\pframe \lb K,t_f \rb =\sum_{n=K}^{t_f} { \binom{t_f}{n} p^n \lb 1-p \rb^{t_f-n}}
\label{Binomial}
\end{equation}
where $p$ is the probability of transmitting a packet on some channel in a single time slot and is given by (\ref{Prob_recursive}) or (\ref{Prob_recursive_Soft}) if the SU's system was overlay or underlay respectively. $\pframe \lb K,t_f \rb$ represents the probability of finding $K$ or more free time slots out of a total of $t_f$ time slots.



In order to guarantee some QoS for the real-time data the SU needs to keep the probability of successful frame transmission above a minimum value denoted $\rmin$, that is $\pframe \geq \rmin$. Hence the problem becomes a throughput maximization problem subject to some average power and QoS constraints as follows
\begin{equation}
\begin{array}{ll}
\rm{maximize}& U_1(\bfGamma(1),\bfP{1})\\
\label{Optim_Prob_Offline}
\rm{subject \; to} &S_1(\bfGamma(1),\bfP{1}) \leq \Pavg\\
& \pframe (K,t_f) \geq \rmin\\
\rm{variables} & \bfGamma(1),\bfP{1}.
\end{array}
\end{equation}
This is the optimization problem assuming an overlay system since we used equations (\ref{Reward}) and (\ref{Average_Power}) for the throughput and power, respectively. It can also be modified systematically to the case of an underlay system. Since there exists a one-to-one mapping between $\pframe(K,t_f)$ and $p$, then there exists a value for $\Dmax$ such that the inequality $p \geq \invDmaxline$ is equivalent to the QoS inequality $\pframe \lb K,t_f \rb \geq \rmin$. That is, we can replace inequality $\pframe(K,t_f) \geq \rmin$ by $p \geq \invDmaxline$ for some $\Dmax$ that depends on $\rmin$, $K$ and $t_f$ that are known a priori. Consequently, problem (\ref{Optim_Prob_Offline}) is reduced to the simpler, yet equivalent, single-time-slot problem (\ref{Prob_Opt_Pow_Control}) and the SU can solve for $\Pst{1}$ and $\Gammast{1}$ vectors following the approach proposed in Section \ref{Problem_Formulation}. The SU solves this problem offline (i.e. before the beginning of the frame transmission) and uses this solution each time slot of the $t_f$ time slots. With this offline scheme, the SU will be able to meet the QoS and the average power constraint requirements as well as maximizing its throughput.

\subsection{Online Power-and-Threshold Adaptation}
In problem (\ref{Prob_Opt_Pow_Control}), we have seen that as $\invDmaxline$ decreases, the system becomes less stringent in terms of the delay constraint. This results in an increase in the average throughput $U_1^*$. With this in mind, let us assume, in the generalized delay model, that at time slot $1$ the SU succeeds in transmitting a packet. Thus, at time slot $2$ the SU has $K-1$ remaining packets to be transmitted in $t_f-1$ time slots. And from the properties of equation (\ref{Binomial}), $\pframe (K-1,t_f-1) > \pframe (K,t_f)$. This means that the system becomes less stringent in terms of the QoS constraint after a successful packet transmission. This advantage appears in the form of higher throughput. To see how we can make use of this advantage, define $\pframe(K(t),t_f-t+1)$ as
\begin{align}
\nonumber &\pframe \lb K(t),t_f-t+1 \rb =\\
&\sum_{n=K(t)}^{t_f-t+1} { \binom{t_f-t+1}{n} \lb p(t) \rb^n \lb 1-p(t) \rb^{t_f-t+1-n}},
\label{Binomial_t}
\end{align}
where $K(t)$ is the remaining number of packets before time slot $t \in \{1,...,t_f\}$ and $p(t)$ is the probability of successful transmission at time slot $t$. At each time slot $\tinset$, the SU modifies the QoS constraint to be $\pframe(K(t),t_f-t+1) \geq \rmin$ instead of $\pframe(K,t_f) \geq \rmin$, that was used in the offline adaptation, and solve the following problem
\begin{equation}
\begin{array}{ll}
\rm{maximize}& U_1(\bfGamma(1),\bfP{1})\\
\label{Optim_Prob_Online}
\rm{subject \; to} &S_1(\bfGamma(1),\bfP{1}) \leq \Pavg\\
& \pframe (K(t),t_f-t+1) \geq \rmin\\
\rm{variables} & \bfGamma(1),\bfP{1},
\end{array}
\end{equation}
to obtain the power and threshold vectors. When the delay constraint in (\ref{Optim_Prob_Online}) is replaced by its equivalent constraint $p \geq \invDmaxline$, the resulting problem can be solved using the overlay approach proposed in Section \ref{Problem_Formulation} without much increase in computational complexity since the power functions and thresholds are given in closed-form expressions. With this online adaptation, the average throughput $U_1^*$ increases while still satisfying the QoS constraint.

\section{Numerical Results}
\label{Results}
We show the performance of the proposed solution for the overlay and underlay scenarios. The slot duration is taken to be unity (i.e. all time measurements are taken relative to the time slot duration), while $\tau=0.05\Ts$. Here, we use $M=10$ channels that suffer i.i.d. Rayleigh fading. The availability probability is taken as $\theta_i=0.05i$ throughout the simulations. The power gain $\gamma$ is exponentially distributed as $f_{\gamma} \left( \gamma \right)=\exp \left( \gamma / \bar{\gamma}\right) / \bar{\gamma}$ where $\bar{\gamma}$ is the average channel gain and is set to be $1$ unless otherwise specified.


Fig. \ref{Overlay_Fig_2_Unconst_Const_NoOSR_M10} plots the expected throughput $U_1^*$ for the overlay scenario after solving problem \eqref{Prob_Opt_Pow_Control}. $U_1^*$ is plotted using equation \eqref{Reward} that represents the average number of bits transmitted divided by the average time required to transmit those bits, taking into account the time wasted due to the blocked time slots. We plot $U_1^*$ with $\Dmax=1.02 \Ts$ and with $\Dmax=\infty$ (i.e. neglecting the delay constraint). We refer to the former problem as constrained problem, while to the latter as unconstrained problem. We also compare the performance to the non optimum stopping rule case (No-OSR) where the SU transmits at the first available channel. We expect the No-OSR case to have the best delay and the worst throughput performances. We can see that the unconstrained problem has the best throughput amongst all constrained problems.

Examining the constrained problem for different sensing orders of the channels, we observe that when the channels are sorted in an ascending order of $\theta_i$, the throughput is higher. This is because a channel $i$ has a higher chance of being skipped if put at the beginning of the order compared to the case if put at the end of the order. This is a property of the problem no matter how the channels are ordered, i.e. this property holds even if all channels have equal values of $\theta_i$. Hence, it is more favorable to put the high quality channels at the end of the sensing order so that they are not put in a position of being frequently skipped. However, this is not necessarily optimum order, which is out of the scope of this work and is left as a future work for this delay-constrained optimization problem.

We also plot the expected throughput of a simple stopping rule that we call $K$-out-of-$M$ scheme, where we choose the highest $K$ channels in availability probability and ignore the remaining channels as if they do not exist in the system. The SU senses those $K$ channels sequentially, probes the gain of each free channel, if any, and transmits on the channel with the highest gain. This scheme has a constant fraction $K\tau/\Ts$ of time wasted each slot. Yet it has the advantage of choosing the best channel among multiple available ones. In Fig. \ref{Overlay_Fig_2_Unconst_Const_NoOSR_M10} we can see that the degradation of the throughput when $K=5$ compared to the optimal stopping rule scheme. The reason is two-fold: 1) Due to the constant wasted fraction of time, and 2) Ignoring the remaining channels that could potentially be free with a high gain if they were considered as opposed to the constrained problem.

The delay is shown in Fig. \ref{Overlay_Fig_3_Unconst_Const_NoOSR_M10} for both the constrained and the unconstrained problems. We see that the unconstrained problem suffers around $6\%$ increase in the delay, at $\Pavg=10$, compared to the constrained one.

\begin{figure}
	\centering
	\includegraphics[width=1\columnwidth]{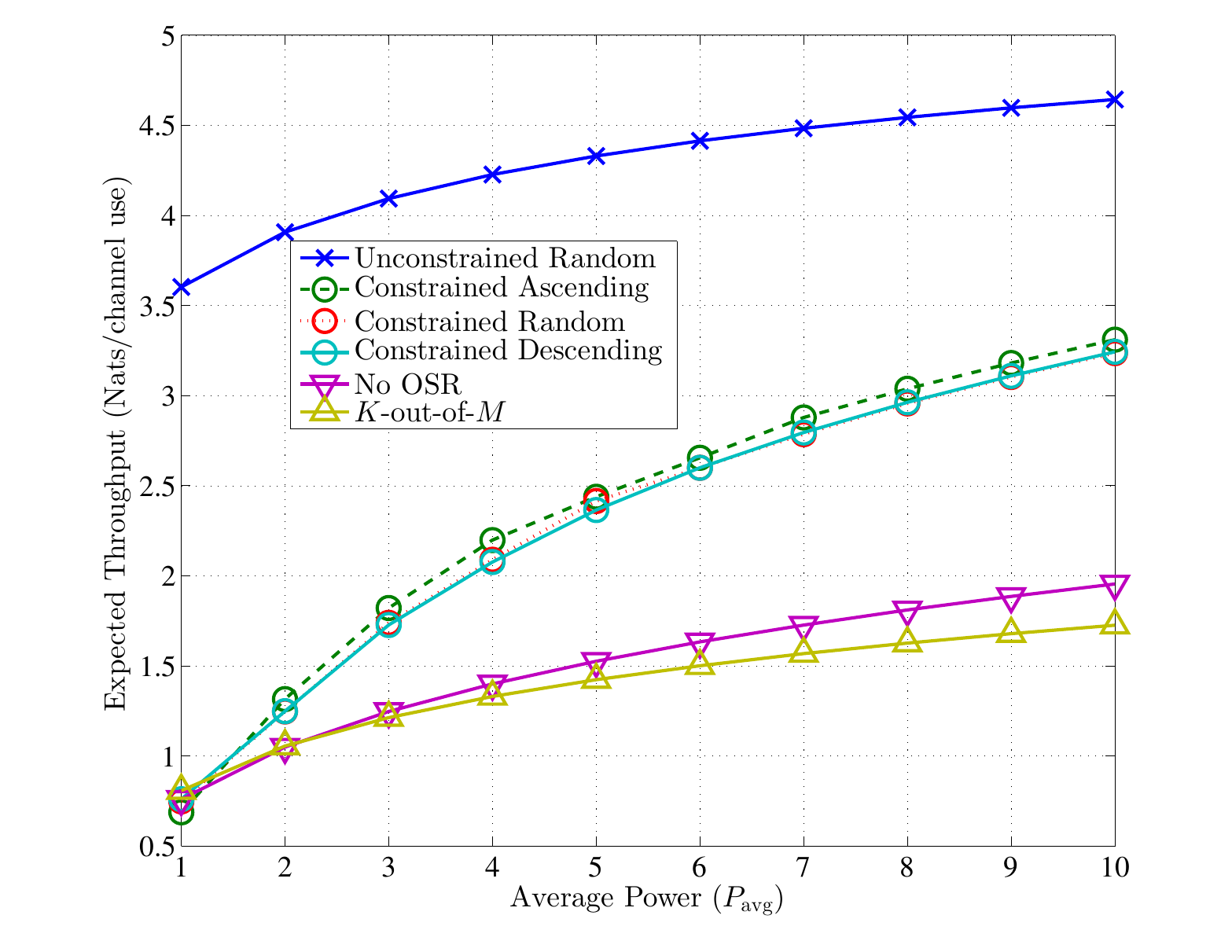}
		\caption{The expected throughput for the overlay scenario for four cases: 1) {\bf Proposed constrained problem}: with average delay constraint for three channel ordering possibilities (ascending ordering of channel availability probabilities, descending ordering, and random ordering), 2) {\bf Unconstrained problem} that ignores the delay constraint, 3) {\bf No optimum stopping rule (No-OSR)} where the SU transmits at the first free channel  and 4) {\bf $K$-out-of-$M$ scheme} where the SU assumes the system has only $K=5$ channels and ignores the remaining $M-K$ channels.}
	\label{Overlay_Fig_2_Unconst_Const_NoOSR_M10}
\end{figure}

\begin{figure}
	\centering
	\includegraphics[width=1\columnwidth]{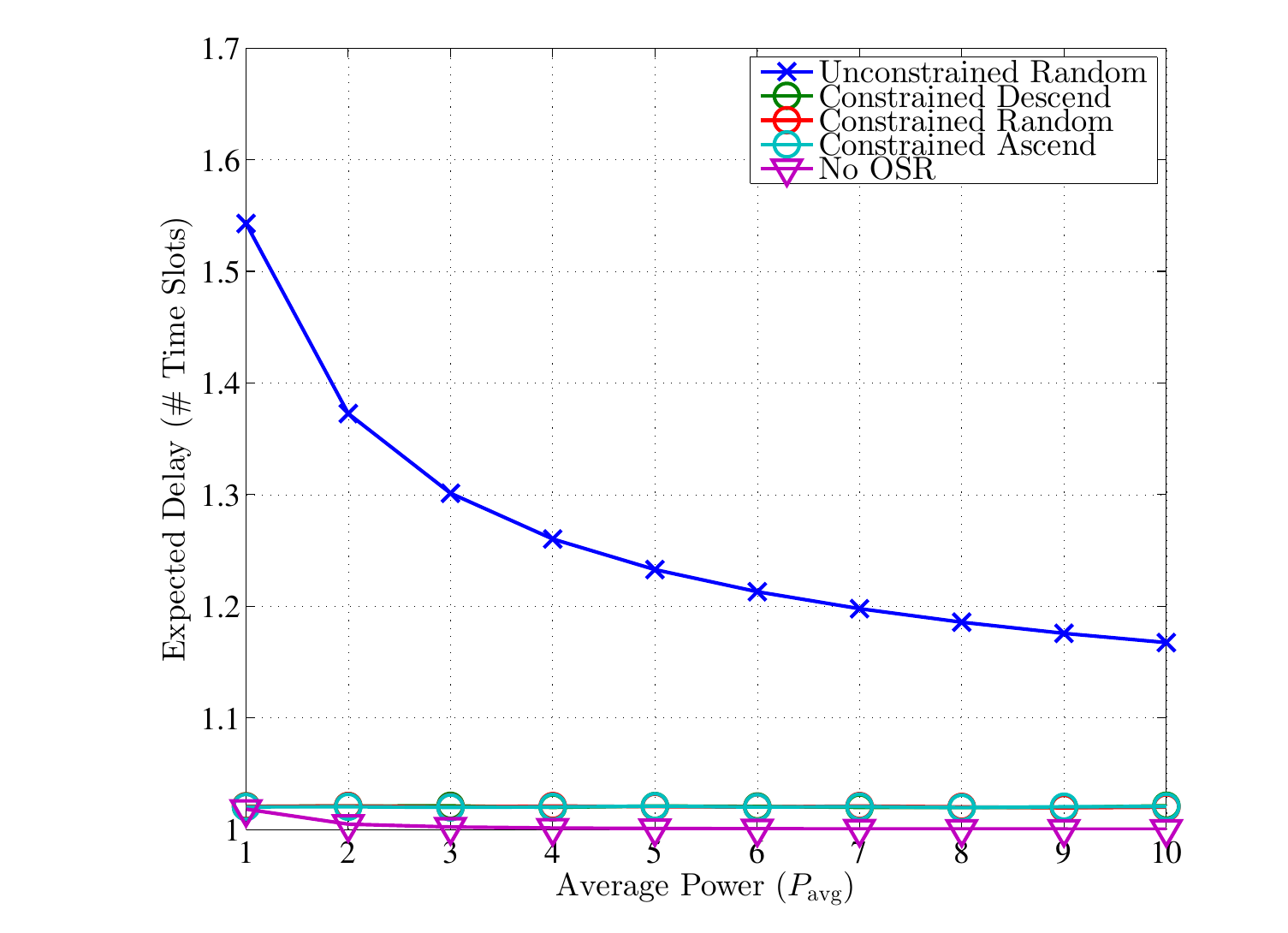}
		\caption{The expected delay for the overlay scenario for problem (\ref{Prob_Opt_Pow_Control}). The unconstrained problem can suffer arbitrary high delay. The constrained problem has a guaranteed average delay for all ordering strategies. The No-OSR scenario, on the other hand, has the best delay performance since the SU uses the first free channel.}
	\label{Overlay_Fig_3_Unconst_Const_NoOSR_M10}
\end{figure}

Studying the system performance under low average channel gain is essential. A low average channel gain represents a SU's channel being in a permanent deep fade or if there is a relatively high interference level at the secondary receiver. Fig. \ref{Overlay_Fig4_Optimum_Threshold} shows $\gammathst{i}$ versus the $\bar{\gamma}$. At low $\bar{\gamma}$, the throughput is expected to be small, hence $\gammathst{i}$ is close to its minimum value $\lambdapst$ so that even if $\gamma_i$ is relatively small, $i$ should not be skipped. In other words, at low average channel gain, the expected throughput is small, thus a relatively low instantaneous gain will be satisfactory for stopping at channel $i$. While when the average channel gain increases, $\gammathst{i}$ should increase so that only high instantaneous gains should lead to stopping at channel $i$. In both cases, high and low $\bar{\gamma}$ there still is a trade-off between choosing a high versus a low value of $\gammathst{i}$.
\begin{figure}[htbp]
	\centering
		\includegraphics[width=1\columnwidth]{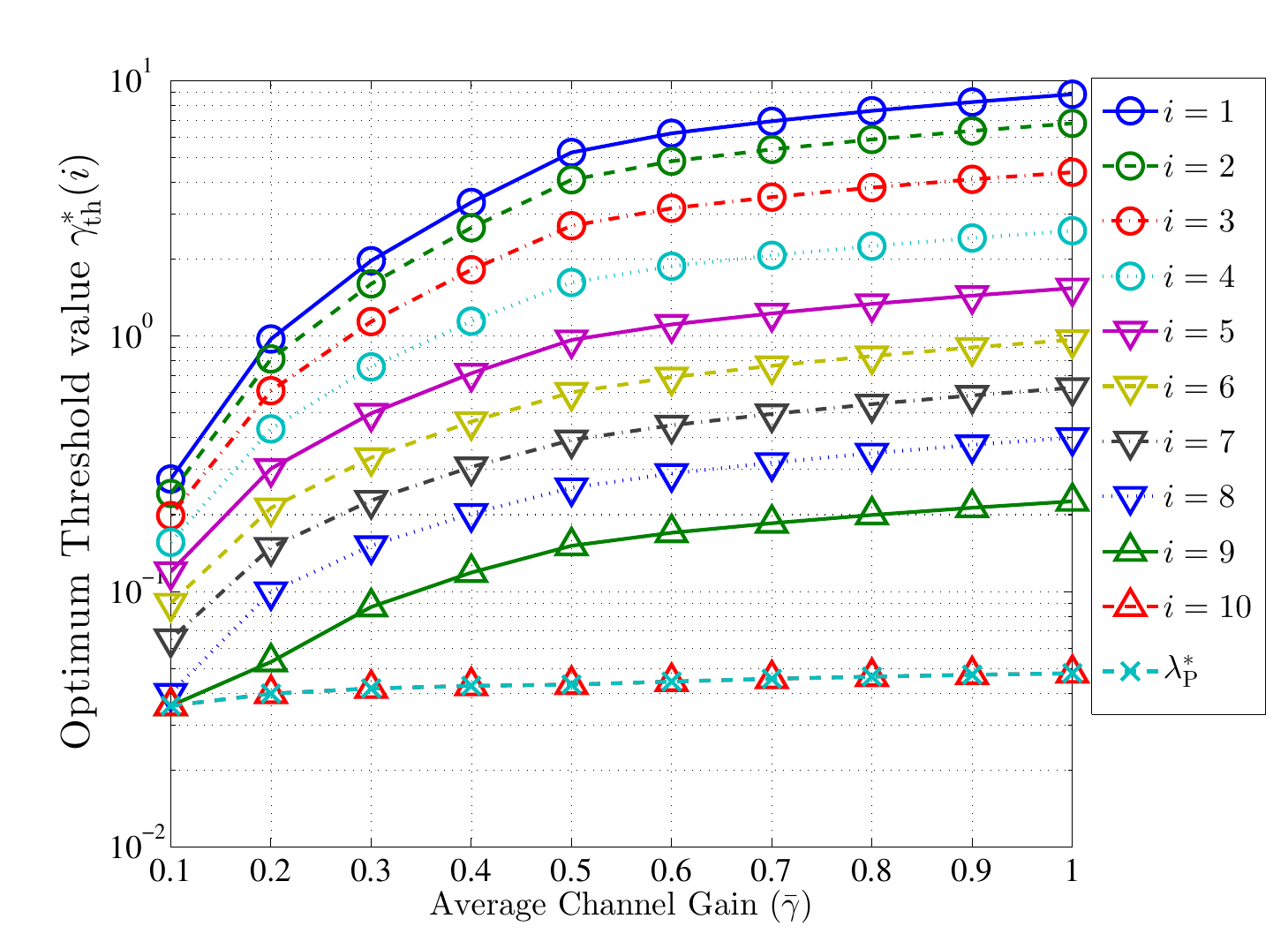}
	\caption{The gap between the optimum threshold $\gammathst{i}$ and its minimum value $\lambdapst$ increases as the average gain increases. This is because as $\bar{\gamma}$ increases, $U_{i+1}$ increases as well. Hence $\gammathst{i}$ increases so that only sufficiently high instantaneous gains should lead to stopping at channel $i$.}
	\label{Overlay_Fig4_Optimum_Threshold}
\end{figure}

The sensing channel (i.e. the channel between the PT and ST over which the ST overhears the PT activity) is modeled as AWGN with unit variance.
 The distributions of the energy detector output $z$ (average energy of N samples sampled from this sensing channel) under the free and busy hypotheses are the Chi-square and a Noncentral Chi-square distributions given by
\begin{align}
\label{X2_Distribution}
&\fzf =\left( \frac{N}{\sigma^2} \right)^N \frac{z^{N-1}}{\left( N-1 \right)!} \exp \left( \frac{-N z}{\sigma^2} \right),\\
\label{NCX2_Distribution}
\nonumber &\fzb =\\
&\left( \frac{N}{\sigma^2} \right) \left( \frac{z}{\cal E} \right)^{\frac{N-1}{2}} \exp \left( \frac{-N \left( z+ {\cal E} \right)}{\sigma^2} \right) \Ibessel{N-1} \left( \frac{2N \sqrt{{\cal E} z}}{\sigma^2} \right),
\end{align}
where $\sigma^2$, which is set to $1$, is the variance of the Gaussian noise of the energy detector, $\cal E$ is the amount of energy received by the ST due to the activity of the PT and is taken as ${\cal E}=2\sigma^2$ throughout the simulations, while $\Ibessel{i}(x)$ is the modified Bessel function of the first kind and $i$th order, and $N=10$.

The main problem we are addressing in this paper is the optimal stopping rule that dictates for the SU when to stop sensing and start transmitting. As we have seen, this is identified by the threshold vector $\Gammasst{1}$. If the SU does not consider the optimal stopping rule problem and rather transmits as soon as it detects a free channel, then it will be wasting future opportunities of possibly higher throughput. Hence, we expect a degradation in the throughput. We plot the two scenarios in Fig. \ref{Underlay_OSR_vs_noOSR_vs_Iavg_M10} for the underlay system with no delay constraint.

\begin{figure}[htbp]
	\centering
		\includegraphics[width=1\columnwidth]{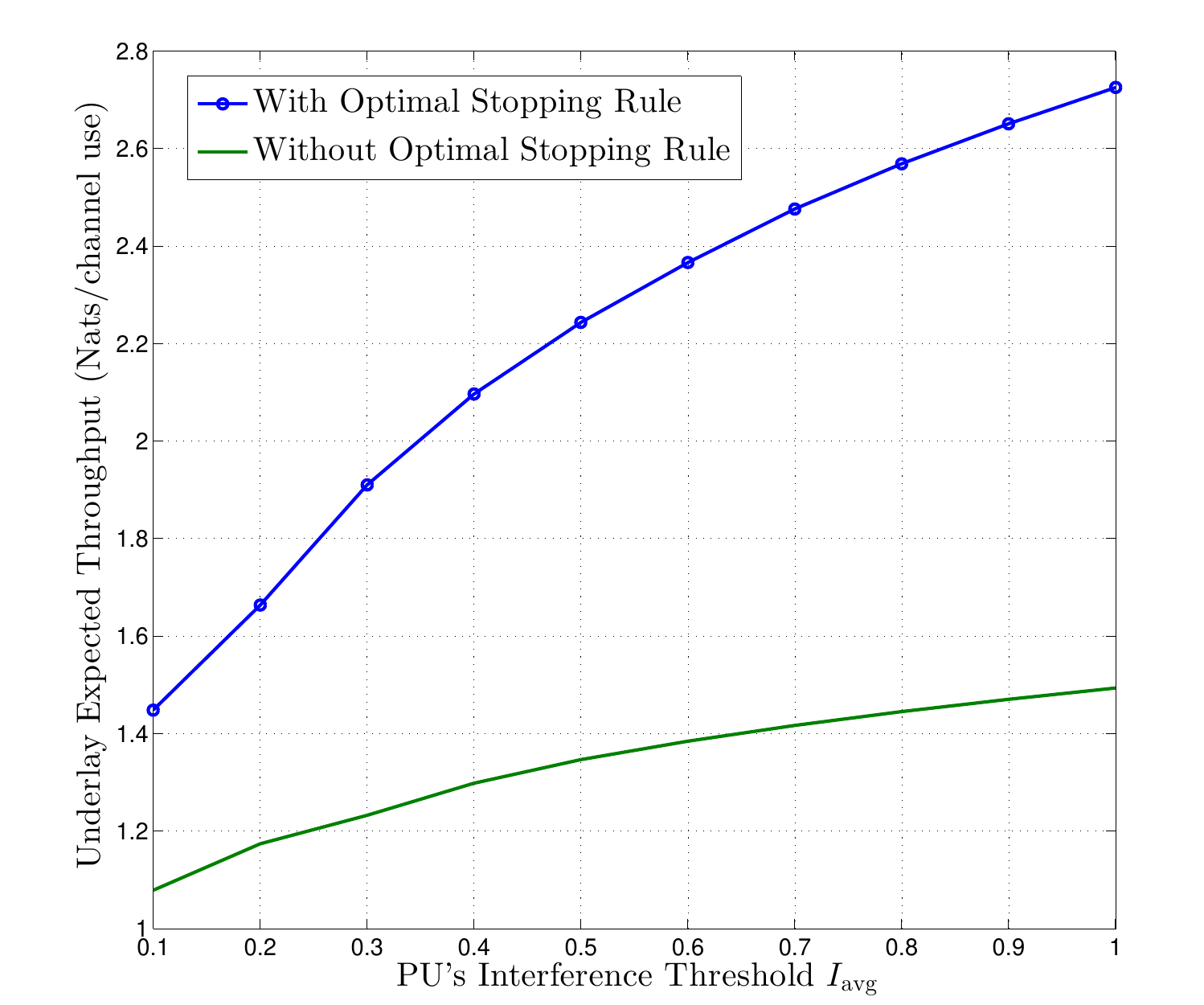}
	\caption{The underlay expected throughput versus the average interference threshold $I_{\rm avg}$. Two scenarios are shown: with and without the optimal stopping rule formulation. In the latter, the SU transmits as soon as a channel is found free.}
	\label{Underlay_OSR_vs_noOSR_vs_Iavg_M10}
\end{figure}

For the multiple SU scenario, the numerical analysis were run for the case of $L=30$ SUs while $M=10$ channels. We assumed the fading channels are i.i.d. among users and among frequency channels. Each channel is exponentially distributed with unity average channel gain. And since $L$ is large, the distribution of the maximum gain among $L$ random gains converges in distribution to the Gumbel distribution \cite{zhang2009asymptotic} having a cumulative distribution function of $\exp \left( -\exp \left( -\gamma/\bar{\gamma}\right)\right)$. The per-user throughput $U_1^{L*}$ is plotted in Fig. \ref{Throughput_MultiSU} where the throughput of the delay-constrained and of the unconstrained optimization problems coincide. This is because when $L\gg M$, the solution of the unconstrained problem is delay optimal as well. Hence, adding a delay constraint does not sacrifice the throughput, when $L$ is large. Moreover, the delay performance shown in Fig. \ref{Delay_MultiSU} shows that the delay does not change with and without considering the average delay constraint since the system is delay- (and throughput-) optimal already.

\begin{figure}%
	\centering
	\includegraphics[width=1\columnwidth]{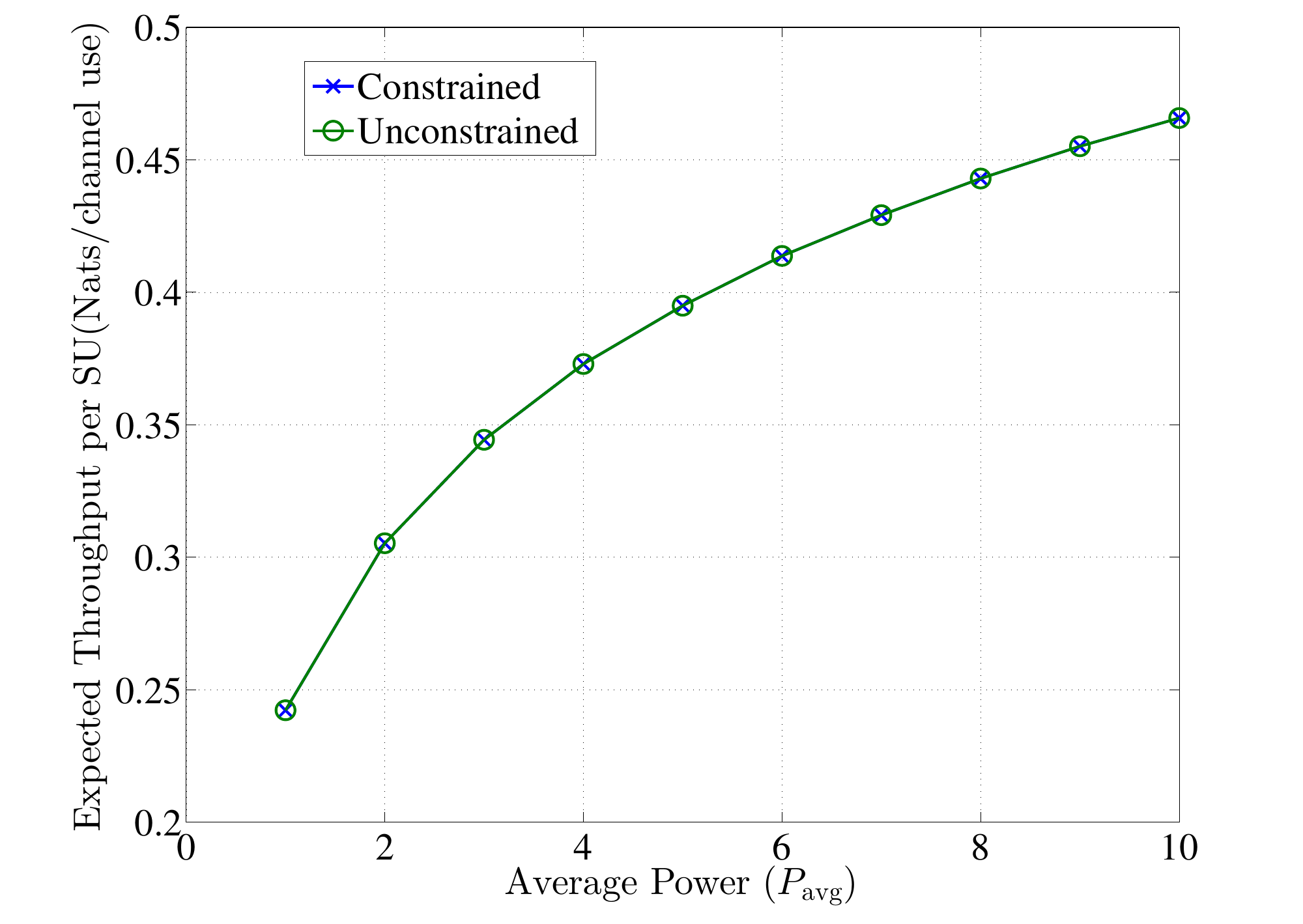}
	\caption{Per user throughput of the system at $L=30$ SUs. The throughput of the constrained and unconstrained problem coincide since the system is throughput (and delay) optimal.}
	\label{Throughput_MultiSU}
	\end{figure}
	\begin{figure}
	\centering
	\includegraphics[width=1\columnwidth]{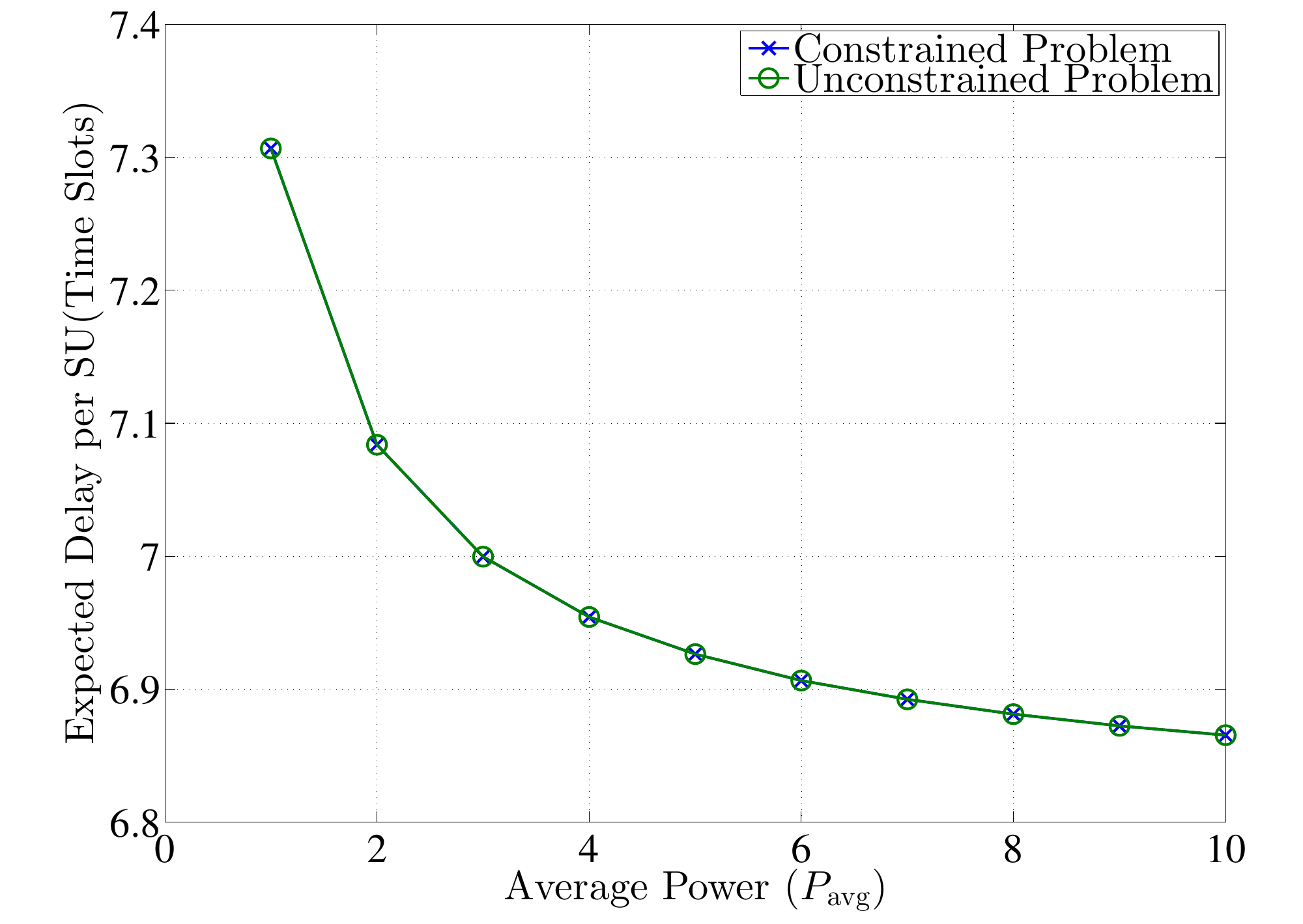}
	\caption{The average delay seen by each user in the system at $L=30$ SUs. The delay of the constrained and unconstrained problems coincide since the system is delay (and throughput) optimal.}
	\label{Delay_MultiSU}
\end{figure}

We have simulated the system for the online algorithm of Section \ref{General_Delay} for $K(1)=2$ packets and $t_f=4$ time slots. We simulated the system at $\rmin=0.95$ which means that the QoS of the SU requires that at least $95\%$ of the frames to be successfully transmitted. Fig. \ref{General_Delay_Overlay_OptPow} shows the improvement in the throughput of the online over the offline adaptation. This is because the SU adapts the power and thresholds at each time slot to allocate the remaining resources (i.e. remaining time slots) according to the remaining number of packets and the desired QoS. This comes at the expense of re-solving the problem at each time slot (i.e. $t_f$ times more).

\begin{figure}%
\centering
\includegraphics[width=1\columnwidth]{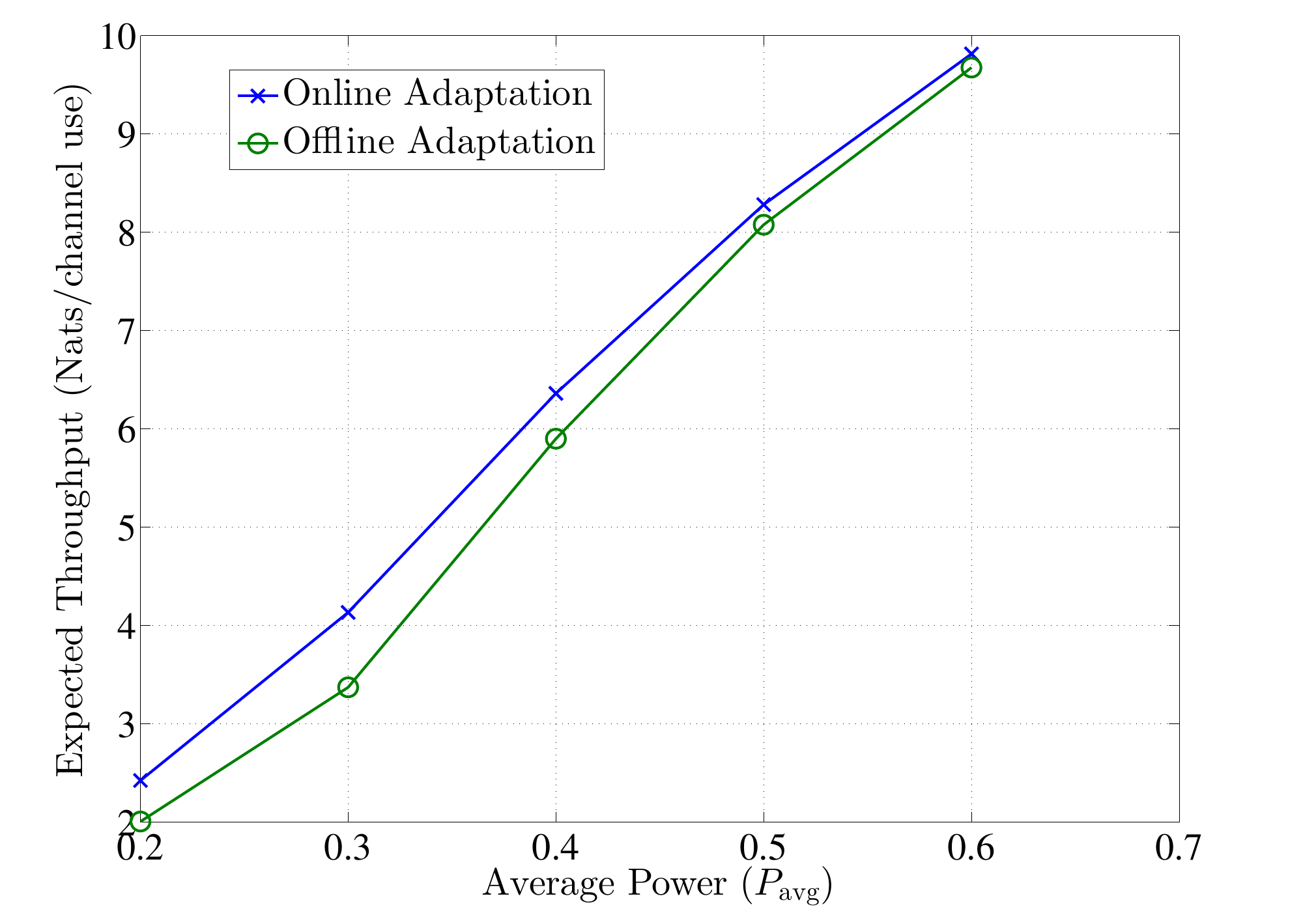}
\caption{The performance of the online adaptation algorithm for the general delay case.}
\label{General_Delay_Overlay_OptPow}
\end{figure}

\section{Conclusion and Discussions}
\label{Conclusion}
In this paper we have formulated the problem of a CR system having a single SU that sequentially tests $M$ potentially available channels, originally licensed to the PU's network, to transmit delay-constrained data. The unique challenge with delay-constrained data is that each packet has a deadline that needs to be transmitted by, on average. Thus there is a trade-off in either taking advantage of a free channel in order to transmit a packet but with low throughput at the current time, or waiting for a future channel that has a considerably higher gain but might be busy. The sequential nature of the problem gave rise to an optimal stopping rule formulation for both overlay and underlay. 

In the overlay scheme, the SU was allowed to transmit on free channels only. We have seen that the optimal power control strategy is solved by a modified version of the water-filling algorithm that takes sequential multiple channels into account. Moreover, the solution of the optimal stopping rule was given explicitly via a set of equations obtained in closed-form expressions. These equations, although derived in a single SU scenario, but are shown to be valid in the multi-SU scenario as well.

In the underlay scheme, on the other hand, the SU is allowed to transmit on any channel even if it was busy as long as the average interference to the PU is tolerable. The solution to the underlay problem involved thresholds that were functions of the sensing instantaneous sufficient statistic. We have provided the optimal closed-form expressions for these thresholds and showed how they depend on the distribution of this sufficient statistic, or more precisely, on the probability of the channel being busy, given this sufficient statistic.

We also discussed the extension of our solution to multiple SUs. We showed that the proposed algorithm can apply to multiple SUs system when the number of SUs is sufficiently larger than the number of channels. Moreover, the optimum solution was found to be throughput optimal and delay optimal at the same time. Our algorithm can reach this solution with a smaller complexity relative to the single SU case.

Finally, our formulations for both the overlay and the underlay incorporated the average delay that the SU's packet experience before being transmitted. We showed that when the average delay is constrained in the optimization problem, we could achieve a relatively low packet delay compared to the delay-unconstrained problem. Then we generalized the problem to consider packets arriving simultaneously and having the same deadline to model typical data. A low-complexity online power-and-threshold adaptation solution was proposed and simulation results showed its performance superiority over the offline solution.

While the problem of finding the optimal sensing order of the channels is outside the scope of this work, one could still rely on previous work that addressed this problem. The work in \cite{Sensing_Order_Poor} proposes a dynamic programming algorithm for this problem but without any delay constraints and, moreover, while fixing $P_i(\gamma)=1$ $\forall i \in \script{M}$. Based on the closed-form expressions of the proposed approach for the thresholds and power functions, one could still find the optimum sensing sequence using this dynamic programming algorithm given some fixed $\lambdap$ and $\lambdad$ (as mentioned at the end of Section \ref{Dual_Variables}). Yet finding the optimum $(\lambdapst,\lambdadst)$ is still an open question. This is because the monotonicity of the $S_1^*$ in $\lambdapst$ is not proven when the sensing order is a variable in the problem. Hence the use of the bisection method in Algorithm \ref{Alg_lambda_P_lambda_D} is not guaranteed to be optimal.

\appendices
\section{Proof of Lemma \ref{Lma_Stochastic_Dominance}}
\label{Apdx_Stochastic_Dominance}
\begin{proof}
We carry out the proof by contradiction. Assume, for some $i$, that $\gamma_{\rm th}^*(i)<\lambdapst$. Thus the reward starting from channel $i$, $U_i \left([\gamma_{\rm th}^*(i), \gamma_{\rm th}^*(i+1),...,\gamma_{\rm th}^*(M)]^T,\Pst{i} \right)$, becomes
\begin{align}
\nonumber & \theta_i c_i \int_{\gamma_{\rm th}^*(i)}^{\infty}{\log(1+P_i^*\gamma) f_{\gamma}(\gamma)} \, d\gamma + \theta_i U_{i+1}^* \int_0^{\gamma_{\rm th}^*(i)}{f_{\gamma}(\gamma)} \, d\gamma \\
&\hspace{1in} + (1-\theta_i)U_{i+1}^*
\label{Stoch_Dom_Definition}
%
%
\\\nonumber \leq & \theta_i c_i \int_{\lambdapst}^{\infty}\hspace{0.2cm}{\log(1+P_i^*\gamma) f_{\gamma}(\gamma)} \, d\gamma + \theta_i U_{i+1}^* \int_0^{\lambdapst}\hspace{0.3cm}{f_{\gamma}(\gamma)} \, d\gamma \\
&\hspace{1in }+ (1-\theta_i)U_{i+1}^*
\label{Stoch_Dom_Inequality}
\\
\label{Stoch_Dom_w_lambda}
=&U_i \left([\lambdapst, \gamma_{\rm th}^*(i+1),...,\gamma_{\rm th}^*(M)]^T,\Pst{i} \right).
\end{align}
Where inequality (\ref{Stoch_Dom_Inequality}) follows by adding the term $\theta_i \lb \int_{\gammathst{i}}^{\lambdapst}{\pdf(\gamma)} \, d\gamma \rb U_{i+1}^*$ to (\ref{Stoch_Dom_Definition}) while (\ref{Stoch_Dom_w_lambda}) follows by the definition of the right-hand-side of (\ref{Stoch_Dom_Inequality}).
Using equation (\ref{Reward}), we can calculate the reward $U_{i-1}$ for both the left-hand-side and right-hand-side of the previous inequality. Thus the following inequality holds
\begin{align}
\nonumber &U_{i-1} \left([\gamma_{\rm th}^*(i-1),\gamma_{\rm th}^*(i),...,\gamma_{\rm th}^*(M)]^T,\Pst{i-1} \right) \leq \\
& U_{i-1} \left([\gamma_{\rm th}^*(i-1), \lambdapst,...,\gamma_{\rm th}^*(M)]^T,\Pst{i-1} \right).
\label{Stoch_Dom_Recursion}
\end{align}
Carrying out the last step recursively $i-2$ more times, we find the relation
\begin{align}
\nonumber &U_1 \left([\gamma_{\rm th}^*(1),...,\gamma_{\rm th}^*(i-1),\gamma_{\rm th}^*(i),...,\gamma_{\rm th}^*(M)]^T,\Pst{1} \right) \leq \\
& U_1 \left([\gamma_{\rm th}^*(1),...,\gamma_{\rm th}^*(i-1),\lambdapst ,...,\gamma_{\rm th}^*(M)]^T,\Pst{1} \right),
\label{Stoch_Dom_Contradiction}
\end{align}
which contradicts with the fact that $\gamma_{\rm th}^*(i)$ is optimal.
\end{proof}

\section{Proof of Theorem \ref{Thm_Unique_Solution_S}}
\begin{proof}
\label{Apdx_Unique_Solution_S}
We first get $S_i^*$, $U_i^*$ and $p_i^*$ by substituting by equations $\gammathst{i}$ and $P_i^*(\gamma)$ in the three equations (\ref{Average_Power}), (\ref{Reward}) and (\ref{Prob_recursive}), respectively. Then we differentiate with respect to $\lambdapst$, treating $\lambdadst$ as a constant, yielding
\begin{align}
\label{Deriv_S_i}
\nonumber \dSdlamp{i}=&-\theta_i \pdf(\gammathst{i}) \dgamdlamp{i} \left( c_i \Pist{i} - S_{i+1}^*\right) -\\
& \theta_i c_i \frac{\ccdf(\gammathst{i})}{\lb \lambdapst \rb^2} + \left( 1 - \theta_i \ccdf(\gammathst{i}) \right) \dSdlamp{i+1},\\
\label{Deriv_U_i}
\nonumber \dUdlamp{i}=&-\theta_i \pdf(\gammathst{i}) \dgamdlamp{i} \times \\
\nonumber &\left[ \lambdapst \lb c_i \Pist{i} - S_{i+1}^* \rb - \lambdadst \lb 1 - p_{i+1}^* \rb \right] - \\
 & \theta_i c_i \frac{\ccdf(\gammathst{i})}{\lambdapst}+ \left( 1 - \theta_i \ccdf(\gammathst{i}) \right) \dUdlamp{i+1},\\
\label{Deriv_p_i}
\dpdlamp{i}=&-\theta_i \pdf(\gammathst{i}) \dgamdlamp{i} \left( 1 - p_{i+1}^* \right) + \\
&\left( 1 - \theta_i \ccdf(\gammathst{i}) \right) \dpdlamp{i+1},
\end{align}
respectively. Multiplying equation (\ref{Deriv_S_i}) by $-\lambdapst$ and equation (\ref{Deriv_p_i}) by $\lambdadst$ then adding them to equation (\ref{Deriv_U_i}) we can easily show that, for all $i \in \script{M}$,
\begin{equation}
\dUSpdlamp{i}=0.
\label{Deriv_U_S_p_i_eq_zero}
\end{equation}

We now find the derivative of $\gammathst{i}$ with respect to $\lambdapst$ by 
differentiating both sides of equation (\ref{gamma_i_Equation}) with respect to $\lambdapst$, while treating $\lambdadst$ as a constant, then using equation (\ref{Deriv_U_S_p_i_eq_zero}), then rearranging we get
\begin{equation}
\dgamdlamp{i} = \frac{c_i \Pist{i}-S_{i+1}^*}{c_i  \frac{\lambdapst}{\gammathst{i}} \Pist{i}}.
\label{gammai_Deriv_Explicit}
\end{equation}
Substituting by equation (\ref{gammai_Deriv_Explicit}) in (\ref{Deriv_S_i}) we get
\begin{align}
\nonumber \dSdlamp{i}=&- \alpha_i \left[ c_i \Pist{i} - S_{i+1}^*\right]^2 -\theta_i c_i \frac{\ccdf(\gammathst{i})}{\lb \lambdapst \rb^2} + \\
&\left( 1 - \theta_i \ccdf(\gammathst{i}) \right) \dSdlamp{i+1},
\label{Deriv_S_i_recursive}
\end{align}
where $\alpha_i$ is given by
\begin{equation}
\alpha_i=\frac{\theta_i \pdf(\gammathst{i})}{c_i  \frac{\lambdapst}{\gammathst{i}} \Pist{i}} \geq 0,
\label{alpha}
\end{equation}
Now evaluating (\ref{Deriv_S_i_recursive}) at $i=M$ and $i=M-1$ we get
\begin{align}
\label{Deriv_S_M}
&\dSdlamp{M}=- \alpha_M \left[ c_M \Pist{M} \right]^2 -\theta_M c_M \frac{\ccdf(\gammathst{M})}{\lb \lambdapst \rb^2},\\
&{\mbox{and }}\nonumber \dSdlamp{M-1} =- \alpha_{M-1} \left[ c_{M-1} \Pist{M-1} -S_M^*\right]^2 \\
\nonumber &\hspace{0.85in}-\theta_{M-1} c_{M-1} \frac{\ccdf(\gammathst{M-1})}{\lb \lambdapst \rb^2}\\
&\hspace{0.85in}+\left( 1 - \theta_{M-1} \ccdf(\gammathst{M-1}) \right) \dSdlamp{M},
\label{Deriv_S_M__1}
\end{align}
respectively. We can see that $\dSdlamp{M} < 0$, hence $\dSdlamp{M-1} < 0$. By induction, let's assume that $\dSdlamp{i+1} < 0$. From (\ref{Deriv_S_i_recursive}) we get that
\begin{align}
\nonumber \dSdlamp{i} =&- \alpha_i \left( c_i \Pist{i} - S_{i+1}^*\right)^2 -\theta_i c_i \frac{\ccdf(\gammathst{i})}{\lb \lambdapst \rb^2} + \\
&\left( 1 - \theta_i \ccdf(\gammathst{i}) \right) \dSdlamp{i+1}<0
\label{Deriv_S_i_induction}
\end{align}
since all its terms are negative. Finally we find that $\dSdlamp{1}<0$ indicating that $S_1^*$ is monotonically decreasing in $\lambdapst$ given any fixed $\lambdadst \geq 0$.

Now, to get an upper bound on $\lambdapst$, we know that
\begin{equation}
S_i^*=\theta_i c_i \int_{\gammathst{i}}^\infty{ \lb \frac{1}{\lambdapst} - \frac{1}{\gamma} \rb \pdf(\gamma) \,d\gamma}+ \left[1-\theta_i \ccdf(\gammathst{i}) \right]S_{i+1}^*.
\label{Average_Power_Opt}
\end{equation}
We can upper bound the first term in (\ref{Average_Power_Opt}) by $\theta_i c_i / \lambdapst$, while $\left[1-\theta_i \ccdf(\gammathst{i}) \right]<1$. Using these two bounds we can write $S_1^* < \sum_{i=1}^M \theta_i c_i / \lambdapst$. But since $S_1^*=\Pavg$, the upper bound on $\lambdapst$, mentioned in Theorem \ref{Thm_Unique_Solution_S}, follows.
\end{proof}

\section{Proof of Lemma \ref{Lma_lambdadst_Bound}}
\label{Apdx_Lma_lambdadst_Bound}
\begin{proof}
We provide a proof sketch for this bound. We know that at the optimal point $p_1^*=\invDmax$ and that $p_1^*=\theta_1 \ccdf \lb\gammathst{1}\rb + \lb 1-\theta_1 \ccdf \lb\gammathst{1}\rb \rb p_2^*$. But since the second term in the latter equation is always positive, then
\begin{align}
\theta_1 \ccdf \lb\gammathst{1}\rb < \invDmax.
\label{Prob_Success_Opt_Inequality}
\end{align}
Substituting by \eqref{Gamma_Solution_Lambert_W} in \eqref{Prob_Success_Opt_Inequality} and rearranging we can upper bound $\lambdadst$ by
\begin{equation}
\nonumber\frac{c_1 \lb {\log {\lb \frac{\lambdapst}{\ccdf^{-1}\lb \frac{1}{\theta_1 \Dmax}\rb}\rb} - {\frac{\lambdapst}{\ccdf^{-1}\lb \frac{1}{\theta_1 \Dmax}\rb}} +1} \rb + U_2^*-\lambdapst S_2^*}{1-p_2^*}
\label{LambdaD_Bound_1}
\end{equation}
We can easily upper bound $\log {\lb \lambdapst/\ccdf^{-1}\lb 1/ \lb\theta_1 \Dmax \rb \rb\rb} - {\lambdapst/\ccdf^{-1}\lb 1/\lb\theta_1 \Dmax \rb\rb}$ by substituting $\lambdapmax$ for $\lambdapst$ when $\lambdapst<\ccdf^{-1}\lb 1/\lb\theta_1 \Dmax \rb\rb$ and by $1$ otherwise. Moreover, it can also be shown that $U_2^*<\Utwomax$, $p_2^*<\ptwomax$ and that $\lambdapst S_2^*>0$ and from Theorem \ref{Thm_Unique_Solution_S} we have $\lambdapst<\lambdapmax$, the proof then follows.
\end{proof}

\bibliographystyle{IEEEbib}
\bibliography{MyLib}

\begin{thebibliography}{10}

\bibitem{Ewaisha_Throughput_Maximization}
A.~Ewaisha and C.~Tepedelenlio\u{g}lu,
\newblock ``{Throughput Maximization} in {Multichannel Cognitive} radio
  {Systems with} {Delay Constraints},''
\newblock in {\em the 47th Asilomar Conference on Signals, Systems, and
  Computers, 2013. IEEE}, November 2013.

\bibitem{Survey_CR_1st_2006_Akyildiz}
Ian~F. Akyildiz, Won-Yeol Lee, Mehmet~C. Vuran, and Shantidev Mohanty,
\newblock ``Next generation/dynamic spectrum access/cognitive radio wireless
  networks: a survey,''
\newblock {\em Comput. Netw.}, vol. 50, no. 13, pp. 2127--2159, Sept. 2006.

\bibitem{POMDP_Qing_Zhao}
Qing Zhao, Lang Tong, Ananthram Swami, and Yunxia Chen,
\newblock ``{Decentralized Cognitive} {MAC} for {Opportunistic Spectrum Access}
  in {Ad Hoc Networks}: {A POMDP Framework},''
\newblock {\em Selected Areas in Communications, IEEE Journal on}, vol. 25, no.
  3, pp. 589 --600, Apr 2007.

\bibitem{Sensing_Order_Poor}
Hai Jiang, Lifeng Lai, Rongfei Fan, and H.V. Poor,
\newblock ``{Optimal Selection} of {Channel} {Sensing Order} in {Cognitive
  Radio},''
\newblock {\em Wireless Communications, IEEE Transactions on}, vol. 8, no. 1,
  pp. 297 --307, Jan. 2009.

\bibitem{Sabharwal_NonCR_Multiband}
A.~Sabharwal, A.~Khoshnevis, and E.~Knightly,
\newblock ``{Opportunistic Spectral Usage}: {Bounds and a Multi-Band} {CSMA/CA
  Protocol},''
\newblock {\em Networking, IEEE/ACM Transactions on}, vol. 15, no. 3, pp.
  533--545, 2007.

\bibitem{Jia_Multichannel_Tx_Opt_Stop_Rule}
Juncheng Jia, Qian Zhang, and Xuemin Shen,
\newblock ``{HC-MAC: A} {Hardware-Constrained Cognitive} {MAC for Efficient}
  {Spectrum Management},''
\newblock {\em Selected Areas in Communications, IEEE Journal on}, vol. 26, no.
  1, pp. 106--117, 2008.

\bibitem{Cheng_Simple_Chan_Sensing}
Ho~Ting Cheng and Weihua Zhuang,
\newblock ``{Simple Channel Sensing} {Order in Cognitive Radio} {Networks},''
\newblock {\em Selected Areas in Communications, IEEE Journal on}, vol. 29, no.
  4, pp. 676--688, 2011.

\bibitem{Pei2013Sensing}
Errong PEI, Jibi LI, and Fang CHENG,
\newblock ``{Sensing-Throughput} {Tradeoff for Cognitive Radio} {Networks with
  Additional} {Primary Transmission Protection}⋆,''
\newblock {\em Journal of Computational Information Systems}, vol. 9, no. 10,
  pp. 3767--3773, 2013.

\bibitem{Adaptive_Rate_Power_CR_Sonia}
V.~Asghari and S.~Aissa,
\newblock ``{Adaptive} {Rate and Power Transmission} in {Spectrum-Sharing
  Systems},''
\newblock {\em Wireless Communications, IEEE Transactions on}, vol. 9, no. 10,
  pp. 3272 --3280, Oct 2010.

\bibitem{Hu_MultiCR_Contention}
Sanqing Hu, Yu-Dong Yao, and Zhuo Yang,
\newblock ``Cognitive medium access control protocols for secondary users
  sharing a common channel with time division multiple access primary users,''
\newblock {\em Wireless Communications and Mobile Computing}, vol. 14, no. 2,
  pp. 284--296, 2014.

\bibitem{Spectrum_Sensing_via_Kolmogorov_Smirnov_Test}
Guowei Zhang, Xiaodong Wang, Ying-Chang Liang, and Ju~Liu,
\newblock ``Fast and robust spectrum sensing via kolmogorov-smirnov test,''
\newblock {\em Communications, IEEE Transactions on}, vol. 58, no. 12, pp.
  3410--3416, December 2010.

\bibitem{Zou2011_CognitiveRelaySelection}
Yulong Zou, Yu-Dong Yao, and Baoyu Zheng,
\newblock ``Cognitive transmissions with multiple relays in cognitive radio
  networks,''
\newblock {\em Wireless Communications, IEEE Transactions on}, vol. 10, no. 2,
  pp. 648--659, 2011.

\bibitem{Robust_Spectrum_Sensing_With_Crowd_Sensors}
Guoru Ding, Jinlong Wang, Qihui Wu, Linyuan Zhang, Yulong Zou, Yu-Dong Yao, and
  Yingying Chen,
\newblock ``Robust spectrum sensing with crowd sensors,''
\newblock {\em Communications, IEEE Transactions on}, vol. 62, no. 9, pp.
  3129--3143, Sept 2014.

\bibitem{Liang2008_SensingThroughputTradeoff}
Ying-Chang Liang, Yonghong Zeng, Edward~CY Peh, and Anh~Tuan Hoang,
\newblock ``Sensing-throughput tradeoff for cognitive radio networks,''
\newblock {\em Wireless Communications, IEEE Transactions on}, vol. 7, no. 4,
  pp. 1326--1337, 2008.

\bibitem{Zou2010_AsymptoticOutageProb}
Yulong Zou, Yu-Dong Yao, and Baoyu Zheng,
\newblock ``Outage probability analysis of cognitive transmissions: impact of
  spectrum sensing overhead,''
\newblock {\em Wireless Communications, IEEE Transactions on}, vol. 9, no. 8,
  pp. 2676--2688, 2010.

\bibitem{Ewaisha2011Optimization}
Ahmed Ewaisha, Ahmed Sultan, and Tamer ElBatt,
\newblock ``Optimization of channel sensing time and order for cognitive
  radios,''
\newblock in {\em Wireless Communications and Networking Conference (WCNC),
  2011 IEEE}. IEEE, 2011, pp. 1414--1419.

\bibitem{Goldsmith_Wireless_Comm}
Andrea Goldsmith,
\newblock {\em Wireless Communications},
\newblock Cambridge University Press, New York, NY, USA, 2005.

\bibitem{Cvx_Boyd}
Stephen Boyd and Lieven Vandenberghe,
\newblock {\em Convex Optimization},
\newblock Cambridge University Press, New York, NY, USA, 2004.

\bibitem{Miersemann_Calc_Var}
Erich Miersemann,
\newblock {\em Calculus of Variations},
\newblock October, 2012.

\bibitem{Lambert_W_Function}
R.~M. Corless, G.~H. Gonnet, D.~E.~G. Hare, D.~J. Jeffrey, and D.~E. Knuth,
\newblock ``{On the Lambert} {W Function},''
\newblock in {\em ADVANCES IN COMPUTATIONAL MATHEMATICS}, 1996, pp. 329--359.

\bibitem{Numerical_Recipes_Ch9}
William~H. Press, Saul~A. Teukolsky, William~T. Vetterling, and Brian~P.
  Flannery,
\newblock {\em Numerical Recipes 3rd Edition: The Art of Scientific Computing},
\newblock Cambridge University Press, New York, NY, USA, 3 edition, 2007.

\bibitem{LinNonlinProg_Luenberger}
David~G. Luenberger and Yinyu Ye,
\newblock {\em Linear and Nonlinear Programming},
\newblock International Series in Operations Research \& Management Science,
  116. Springer, New York, third edition, 2008.

\bibitem{FCC_Interference_Threshold}
David Furth, Diane Conley, Chris Murphy, Bruce Romano, and Scot Stone,
\newblock ``Federal communications commission spectrum policy task force report
  of the spectrum rights and responsibilities working group,''
\newblock {\em Spectrum}, 2002.

\bibitem{Spectrum_Sensing_Survey_Poor}
E.~Axell, G.~Leus, E.G. Larsson, and H.V. Poor,
\newblock ``Spectrum sensing for cognitive radio : State-of-the-art and recent
  advances,''
\newblock {\em Signal Processing Magazine, IEEE}, vol. 29, no. 3, pp. 101--116,
  2012.

\bibitem{Spectrum_Sensing_Survey}
T.~Yucek and H.~Arslan,
\newblock ``A survey of spectrum sensing algorithms for cognitive radio
  applications,''
\newblock {\em Communications Surveys Tutorials, IEEE}, vol. 11, no. 1, pp. 116
  --130, quarter 2009.

\bibitem{Wireless_Comm_Tse}
David Tse and Pramod Viswanath,
\newblock ``Fundamentals of wireless communications,'' 2004.

\bibitem{Adaptive_NC_Deadline}
Lei Yang, Yalin~Evren Sagduyu, and Jason~Hongjun Li,
\newblock ``Adaptive network coding for scheduling real-time traffic with hard
  deadlines,''
\newblock in {\em Proceedings of the thirteenth ACM international symposium on
  Mobile Ad Hoc Networking and Computing}, New York, NY, USA, 2012, MobiHoc
  '12, pp. 105--114, ACM.

\bibitem{zhang2009asymptotic}
Yuan Zhang et~al.,
\newblock ``Asymptotic ber analysis of a simo multiuser diversity system,''
\newblock {\em IEEE Transactions on Vehicular Technology}, vol. 58, no. 9, pp.
  5330--5335, 2009.

\end{thebibliography}
\begin{IEEEbiography}[{\includegraphics[width=1in,height=1.25in,clip,keepaspectratio]{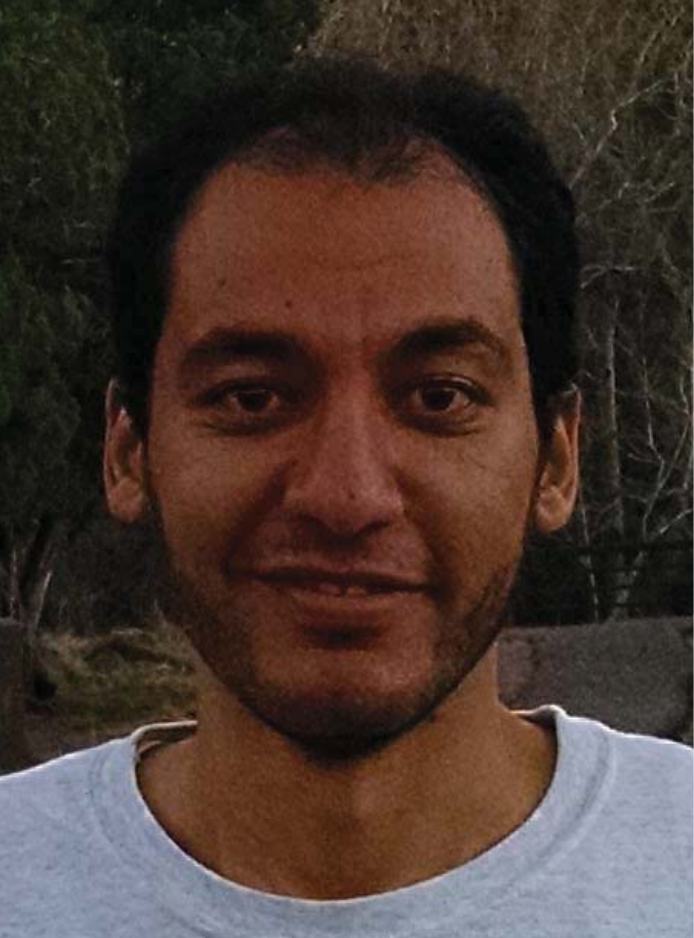}}]{Ahmed E. Ewaisha} was born in Cairo, Egypt in 1987. He received his B.S. degree with honors in electrical engineering ranking top 5\% on his class at Alexandria University in 2009. Consequently, he was admitted to Nile University that is considered the first research-based university in Egypt where he received his M.S. degree in 2011 in electrical engineering.

In fall 2011, he joined the Ira A. Fulton School of engineering at Arizona State University, Tempe, where he is now working towards his PhD degree studying the delay analysis in cognitive radio networks. His research interests span a wide area of wireless as well as wired communication networks including stochastic optimization, power allocation, cognitive radio networks, resource allocation and quality-of-service guarantees in data networks.
\end{IEEEbiography}

\begin{IEEEbiography}[{\includegraphics[width=1in,height=1.25in,clip,keepaspectratio]{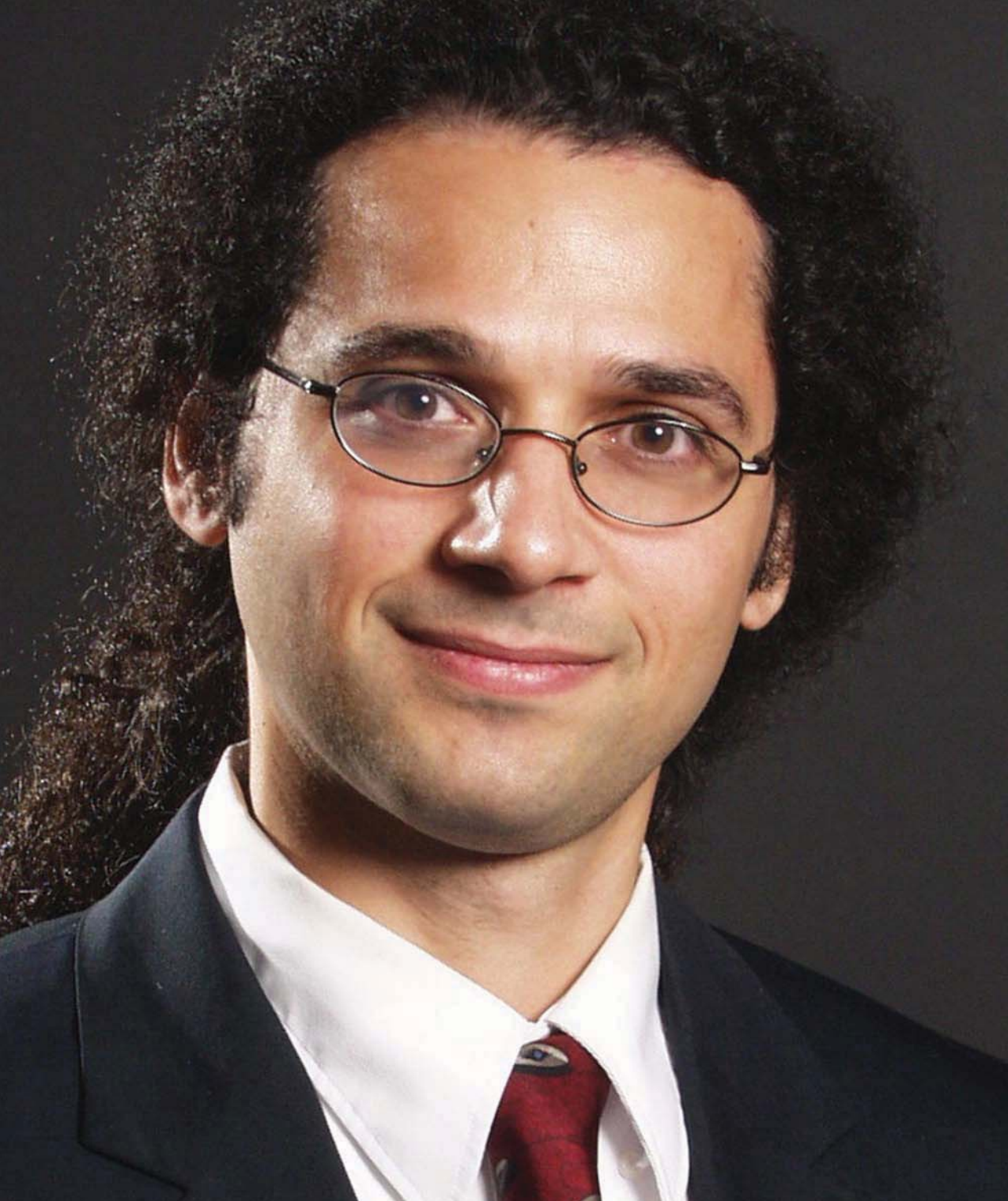}}]{Cihan Tepedelenlio\u{g}lu} (S'97-M'01) was born in Ankara, Turkey in 1973. He received his B.S. degree with highest honors from Florida Institute of Technology in 1995, and his M.S. degree from the University of Virginia in 1998, both in electrical engineering. From January 1999 to May 2001 he was a research assistant at the University of Minnesota, where he completed his Ph.D.  degree in electrical and computer engineering. He is currently an Associate Professor of Electrical Engineering at Arizona State University.

Prof. Tepedelenlio\u{g}lu was awarded the NSF (early) Career grant in 2001, and has served as an Associate Editor for several IEEE Transactions including IEEE Transactions on Communications, IEEE Signal Processing Letters, and IEEE Transactions on Vehicular Technology. His research interests include statistical signal processing, system identification, wireless communications, estimation and equalization algorithms for wireless systems, multi-antenna communications, OFDM, ultra-wideband systems, distributed detection and estimation, and data mining for PV systems.
\end{IEEEbiography}

\end{document}